\documentclass[conference]{IEEEtran}
\IEEEoverridecommandlockouts
\usepackage{graphicx}
\usepackage{multirow}
\usepackage{amsmath,amssymb,amsfonts}
\usepackage{amsthm}
\usepackage{mathrsfs}
\usepackage{xcolor}
\usepackage{textcomp}
\usepackage{manyfoot}
\usepackage{booktabs}
\usepackage{algorithm}
\usepackage{algorithmicx}
\usepackage{algpseudocode}
\usepackage{listings}

\theoremstyle{thmstyleone}
\newtheorem{theorem}{Theorem}

\newtheorem{claim}{Claim}
\theoremstyle{thmstyletwo}%

\theoremstyle{thmstylethree}%
\newtheorem{definition}{Definition}%
\newtheorem{lemma}{Lemma}%

\usepackage{xcolor}
\usepackage{comment}
\usepackage{tikz}

\usepackage[listings]{tcolorbox}
\usepackage{xspace}
\tcbuselibrary{all}
\usepackage{pifont}
\usepackage{mathtools}
\usepackage{pgfplots}
\pgfplotsset{compat=newest}

\usepackage{paralist}
\usepackage{tabularx}
\usepackage{booktabs}
\usepackage{bookmark}
\usepackage{diagbox}
\usepackage{multirow}
\usepackage{hhline}

\newcommand{\IDi}{\ensuremath{\mathsf{id}_i}\xspace}

\newcommand{\abort}{\ensuremath{\mathsf{abort}}\xspace}
\newcommand{\initok}{\ensuremath{\mathsf{OK}}\xspace}

\newcommand{\FF}{\mathcal{F}}

\newcommand{\Setup}{\ensuremath{\mathsf{Setup}}\xspace}
\newcommand{\Generate}{\ensuremath{\mathsf{Generate}}\xspace}
\newcommand{\Prove}{\ensuremath{\mathsf{Prove}}\xspace}
\newcommand{\Reveal}{\ensuremath{\mathsf{Reveal}}\xspace}

\newcommand{\model}{\ensuremath{M}\xspace}
\newcommand{\publicwitness}{\ensuremath{\mathsf{tk}}\xspace}
\newcommand{\globalpublicwitness}{\ensuremath{\mathsf{tk}_g}\xspace}
\newcommand{\secretwitness}{\ensuremath{\mathsf{w}}\xspace}
\newcommand{\globalwitness}{\ensuremath{\mathsf{W}_g}\xspace}
\newcommand{\proofstate}{\ensuremath{\mathsf{pf}}\xspace}
\newcommand{\prfvalue}{\ensuremath{\mathsf{R}}\xspace}
\newcommand{\sakeyk}{\ensuremath{\mathsf{k}_i}\xspace}
\newcommand{\sarandb}{\ensuremath{\mathsf{b}_i}\xspace}
\newcommand{\sapairwisemask}{\ensuremath{\mathsf{k}_{i,j}}\xspace}
\newcommand{\pairwisemaskext}{\ensuremath{\mathsf{g^{{k_i}{k_j}}}}\xspace}
\newcommand{\sharekey}{\ensuremath{\mathsf{s}^k_{i,j}}\xspace}
\newcommand{\sharerandom}{\ensuremath{\mathsf{s}^b_{i,j}}\xspace}
\newcommand{\saclientspublic}{\ensuremath{\mathsf{g}^{k_i}}\xspace}
\newcommand{\saneighborpublic}{\ensuremath{\mathsf{g}^{k_j}}\xspace}
\newcommand{\saneighborkeyk}{\ensuremath{\mathsf{k}_j}\xspace}
\newcommand{\clientkeymaterials}{\ensuremath{\mathsf{key}_{c_i}}\xspace}
\newcommand{\serverkeymaterials}{\ensuremath{\mathsf{key}_s}\xspace}

\newcommand{\clientset}{\ensuremath{\mathcal{C}}\xspace}
\newcommand{\clienti}{\ensuremath{C_i}\xspace}
\newcommand{\onlineclientset}{\ensuremath{\mathcal{C}_o}\xspace}
\newcommand{\droppedclientset}{\ensuremath{\mathcal{C}_d}\xspace}
\newcommand{\flserver}{\ensuremath{S_{fl}}\xspace}
\newcommand{\posserver}{\ensuremath{S_{sp}}\xspace}

\newcommand{\G}{\ensuremath{\mathbb{G}}\xspace}

\newcommand{\param}{\ensuremath{\mathsf{pp}}\xspace}
\newcommand{\mlparam}{\ensuremath{\mathsf{pp}_{ml}}\xspace}
\newcommand{\tsparam}{\ensuremath{\mathsf{pp}_{ts}}\xspace}

\newcommand{\Zp}{\ensuremath{\mathbb{Z}_p}}
\newcommand{\g}{\ensuremath{\mathsf{g}}\xspace}
\newcommand{\Zq}{\ensuremath{\mathbb{Z}_q}}

\newcommand{\allpk}{\ensuremath{\mathcal{VK}}\xspace}

\newcommand{\negl}{\ensuremath{\mathsf{negl}(\lambda)}\xspace}
\newcommand{\poly}{\ensuremath{\mathsf{poly}}\xspace}

\newcommand{\hash}{\ensuremath{\mathsf{H}}\xspace}
\newcommand{\hashone}{\ensuremath{\mathsf{H}_1}\xspace}
\newcommand{\hashtwo}{\ensuremath{\mathsf{H}_2}\xspace}
\newcommand{\Train}{\ensuremath{\mathsf{Train}}\xspace}
\newcommand{\SAGen}{\ensuremath{\mathsf{SASetup}}\xspace}
\newcommand{\SAProtect}{\ensuremath{\mathsf{SAProtect}}\xspace}
\newcommand{\SAAgg}{\ensuremath{\mathsf{SAAggregate}}\xspace}

\newcommand{\prfkey}{\ensuremath{\mathsf{K}}\xspace}

\newcommand{\idata}{\ensuremath{D_i}\xspace}
\newcommand{\paramsa}{\ensuremath{\mathsf{pp}_{sa}}\xspace}
\newcommand{\mlmsg}{\ensuremath{\mathsf{m}_{i}}\xspace}
\newcommand{\mlmsgprotected}{\ensuremath{\mathsf{c}_{i}}\xspace}

\newcommand{\pk}{\ensuremath{\mathsf{vk}}\xspace}
\newcommand{\sk}{\ensuremath{\mathsf{sk}}\xspace}
\newcommand{\pki}{\ensuremath{\mathsf{vk}_i}\xspace}
\newcommand{\ski}{\ensuremath{\mathsf{sk}_i}\xspace}

\newcommand{\msg}{\ensuremath{\mathsf{msg}}\xspace}
\newcommand{\sig}{\ensuremath{\mathsf{\sigma}}\xspace}

\newcommand{\TSKeyGen}{\ensuremath{\mathsf{TSKeyGen}}\xspace}
\newcommand{\TSSign}{\ensuremath{\mathsf{TSSign}}\xspace}
\newcommand{\TSSignAgg}{\ensuremath{\mathsf{TSSignAgg}}\xspace}
\newcommand{\TSVerify}{\ensuremath{\mathsf{TSVerify}}\xspace}

\newcommand{\Adv}{\ensuremath{\mathcal{A}}\xspace}
\newcommand{\B}{\ensuremath{\mathcal{B}}\xspace}

\newcommand{\Sim}{\ensuremath{\mathcal{S}}\xspace}
\newcommand{\FedPoP}{\ensuremath{\mathsf{FedPoP}}\xspace}

\newcommand{\Dist}{\ensuremath{\mathcal{S}}\xspace}
\newcommand{\Real}{{\mathsf{REAL}}}
\newcommand{\Ideal}{{\mathsf{IDEAL}}}
\newcommand{\Func}[1]{{\FF_{\scriptstyle\mathsf{#1}}}}
\newcommand{\Prot}[1]{\ensuremath{\Pi_{\scriptstyle\mathsf{#1}}}}
\newcommand{\Fpop}{\ensuremath{\Func{FedPoP}}\xspace}
\newcommand{\Ppop}{\ensuremath{\Prot{FedPoP}}\xspace}

\newtcolorbox[auto counter,number format=\Roman ]{functionality}[2][]{enhanced,colback=white,
fonttitle=\bfseries,coltitle=gray!25!black,
attach boxed title to top left=
{xshift=2mm,yshift=-3mm,yshifttext=-1mm},
boxed title style={colframe=gray!75!black,
colback=yellow!50!gray},
title=Functionality #2,#1}
\newtcolorbox[auto counter,number format=\Roman ]{scheme}[2][]{enhanced,colback=white,
fonttitle=\bfseries,coltitle=gray!25!black,
attach boxed title to top left=
{xshift=2mm,yshift=-3mm,yshifttext=-1mm},
boxed title style={colframe=gray!75!black,
colback=green!10!gray},
title=#2,#1}

\definecolor{plotblue}{RGB}{0,114,189}
\definecolor{plotred}{RGB}{217,83,25}
\definecolor{plotgreen}{RGB}{119,172,48}
\definecolor{plotpurple}{RGB}{126,47,142}

\begin{document}

\title{\FedPoP: Federated Learning Meets Proof of Participation}

\author{
\IEEEauthorblockN{
Devriş İşler\IEEEauthorrefmark{1}\IEEEauthorrefmark{2}, 
Elina van Kempen\IEEEauthorrefmark{3}, 
Seoyeon Hwang\IEEEauthorrefmark{4}, 
and Nikolaos Laoutaris\IEEEauthorrefmark{1}
}

\IEEEauthorblockA{
\IEEEauthorrefmark{1}IMDEA Networks Institute, Spain\\
\IEEEauthorrefmark{2}Universidad Carlos III de Madrid, Spain\\
\IEEEauthorrefmark{3}University of California, Irvine, USA\\
\IEEEauthorrefmark{4}Stealth Software Technologies Inc., USA\\
Email: devris.isler@gmail.com, evankemp@uci.edu, seoyeon@stealthsoftwareinc.com, nikolaos.laoutaris@imdea.org
}

\thanks{Corresponding author: Devriş İşler (devris.isler@gmail.com). This version is currently under review.}
}
\maketitle
\begin{abstract} 
Federated learning (FL) offers privacy preserving, distributed machine learning, allowing clients to contribute to a global model without revealing their local data.  
As models increasingly serve as monetizable digital assets, the ability to prove participation in their training becomes essential for establishing ownership. 
In this paper, we address this emerging need by introducing \FedPoP,
a novel FL framework that allows non-linkable proof of participation while preserving client anonymity and privacy without requiring either extensive computations or a public ledger.
\FedPoP is designed to seamlessly integrate with existing secure aggregation protocols to ensure compatibility with real-world FL deployments.
We provide a proof of concept implementation and an empirical evaluation under realistic client dropouts. In our prototype, \FedPoP introduces 0.97 seconds of per-round overhead atop securely aggregated FL and enables a client to prove its participation/contribution to a model held by a third party in 0.0612 seconds. 
These results indicate \FedPoP is practical for real-world deployments that require auditable participation without sacrificing privacy.
\end{abstract}


\section{Introduction}\label{intro}
Federated learning (FL) \cite{mcmahan2017communication} has become one of the innovative distributed machine learning structures wherein private data holders (a.k.a. \textit{clients}) contribute to a global model  initialized and aggregated by a (federated learning) \textit{server} (so-called \textit{aggregator}) 
without revealing their raw input data. 
The most common FL setting involves three parties: a \textit{server} who initiates a model and aggregates training data (local models) from clients,  
a large number of \textit{clients} who collaboratively train the model, and 
a \textit{service provider} who deploys the model to provide services to its users. 
In a nutshell, an FL system consists of iterative aggregation rounds where 1) the server sends global model parameters to clients; 2) each client trains the model using its own private data and transmits updated parameters to the server; and 3) the server aggregates the updated parameters sent by the clients into a new global model using an aggregation procedure (e.g., FedAvg \cite{mcmahan2017communication}, FedQV \cite{fedqvtian24}). 
The final global model is delivered to the service provider when the training is completed. 
Unlike traditional machine learning that relies on centralized data collection for model training,  
FL inherently preserves client data privacy by transmitting only model updates to the server instead of exposing their raw data. 
It is further advanced with secure aggregation protocols~\cite{flsok23} to prevent the server from inferring private information from individual model updates, e.g., via privacy attacks \cite{FLvulner21}.

Due to their powerful capabilities derived from training with massive client data, FL models have been increasingly adopted in various large-scale applications from improving virtual keyboard search (e.g., Google Keyboard Gboard~\cite{googlekeyboard}) to prediction of COVID-19 patients' future oxygen requirements~\cite{flcovidprediction}. 
These trends position FL models as monetizable assets \cite{mlaaSRibeiroGC15}, providing a means for clients to extract economic benefit from their data.
However, this monetization raises the issue of preserving ownership and protecting the intellectual property rights of the models. 
For instance, consider a service provider monetizes a trained model in unauthorized or unintended ways. In such cases, it becomes crucial for clients (as well as servers) to have a means of demonstrating their participation and ownership, for example, in pursuing legal recourse.
 
Note that in the context of FL, participation in training provides the primary basis for ownership claims, as the global model is derived from the aggregation of client updates.
This type of proof is important not only for establishing economic benefits, but also for enforcing their data rights under data protection regulations such as the General Data Protection Regulation~\cite{gdpr} or the California Consumer Privacy Act~\cite{ccpa}, which grant individuals the ability to access, update, or delete their data. 
For instance, a client may request the removal of their contribution from a deployed model~\cite{LiuXYWL22forgettonfl}. 
To comply, the service provider must first determine whether the client participated in the training process. 
Therefore, mechanisms for verifying participation in FL become a prominent feature to protect ownership and regulatory compliance.

Prior work on proving the participation/ownership is grounded on either 1) \textit{watermarking}; or 2) \textit{blockchain}. 
In the watermarking approach, a watermark is embedded on the client side or the server side. 
In a client-side watermarking \cite{fedipr23,fedcip23,waffle21,flwm22,FedZKP}, clients insert a watermark to the model during the local training while server-side watermarking \cite{flareTekgulA23} (so-called \textit{fingerprinting}) focuses on detecting intellectual property infringement  to learn which client unlawfully distributed the model. 
On the other hand, blockchain approaches \cite{BuyukatesHHFZLFA23,pflm21} utilize blockchain and more advanced cryptographic solutions (e.g., zero-knowledge proofs) to ensure integrity and verifiability in FL while the identities of the clients are stored on the public ledger. 
\cite{BuyukatesHHFZLFA23} proposes a data marketplace of privacy-preserving FL providing proof-of contribution based reward allocation so that the clients are compensated based on their contributions to the model, and zero-knowledge proofs are used to ensure verifiability of all the computations in their protocol (e.g., clients' contribution assessment).  
\cite{pflm21}, on the other hand, provides the proof of participation functionality to protect clients against deceiving attacks where the server deceives a subset of clients into thinking that the other clients are offline. 
\cite{pflm21} mitigates the deceiving attack by enforcing the server to announce the online clients over blockchain. 
While the aforementioned approaches provide a proof of participation functionality on the top of FL, they have the following three shortcomings: 1) the identities of clients are public (\textit{privacy} and \textit{anonymity}); 2) the model may be needed to perform proof of participation (\textit{model privacy}); and 3) a curious party can derive statistics about a client's participation over a period of time (\textit{unlinkability}). 
For example, in \cite{pflm21}, \textit{anyone} can gain the knowledge of the clients' identities because of the public ledger, drive statistics about clients' participation for a given time, and link a client to its multiple participation. 

To overcome these limitations, 
we propose \FedPoP, a novel federated learning framework that allows privacy-preserving proof of participation without relying on public ledgers or expensive computations.  
In \FedPoP, a client can prove its participation in a global model to the service provider in a secure and private manner. 
\FedPoP ensures that no one can forge a valid proof of participation (\textit{soundness}), no information about other participants is disclosed (\textit{privacy}), the identity of the proving client remains hidden (\textit{anonymity}), and multiple proofs from the same client cannot be linked (\textit{unlinkability}).
Our protocol is based on secure aggregation and lightweight cryptographic primitives, including threshold signatures and oblivious pseudorandom functions, to guarantee these properties. 

\textbf{Our approach.} 
In our approach, clients produce their local updates using the model parameters sent by a server and hide their local updates before sending them to the server via secure aggregation as in secure FL. 
The server securely aggregates the local updates and returns the global model to the clients. 
With the help of the server, the (participating) clients to jointly generate a high entropy secret as a secret witness.  
Later, the clients, using the server as a proxy, jointly generate a proof of participation.
Each party in the system receives also a global witness. 
During the proof of participation, a suspicious client communicates with a service provider. 
The client first proves that it has a proof of participation (e.g., the signature) verified under the global public witness. 
Second, the client proves that it knows the secret witness as well. 
The client does not reveal any information about its secret witness to the service provider.  
Furthermore, \FedPoP ensures that the client's identity is kept secret and the proof of participation does not reveal the identity of the other clients. 
The service provider cannot link the multi-participation of the same client to the model. 
A malicious adversary has to corrupt a certain number of clients in order to forge a signature. 
\FedPoP utilizes threshold signatures to enable clients to generate a signature on a global model even in the presence of dropped-out clients and (oblivous) pseudorandom function to prove the knowledge of the secret witness without necessitating expensive computation such as zero-knowledge proof.
\FedPoP is effective regardless of the underlying secure aggregation protocol since it makes no modification to the aggregation protocol. 

We implemented \FedPoP as a proof-of-concept and evaluated it with various parameters and thresholds. 
Our evaluation results show that \FedPoP introduces 0.97 seconds of overhead on the top of a securely aggregated FL, and that a client can prove its participation in only $0.0612$ seconds. 

\textbf{Our contributions are summarized as follows:} 
\begin{itemize}
   \item We identify the need of a new privacy-preserving proof of participation in FL and define its security and functionality requirements. 
    \item We propose \FedPoP, the first FL framework that provides proof of participation with guarantees of privacy, anonymity, and unlinkability.
    \item We implement and evaluate \FedPoP, showing its feasibility with minimal overhead atop a securely aggregated FL settings. 
\end{itemize}

\section{Related Work}\label{relatedwork}
The strategies proposed by prior work on safeguarding model ownership and verifying it can be categorized into two: 
1) \textit{blockchain}; 
and 2) \textit{watermarking}.

\textit{\textbf{Blockchain-based methods.}}
One line of research has leveraged public ledgers, i.e. blockchain technologies, to (in)directly provide proof of participation in a federated learning setting. 
For instance, Jiang et al. \cite{pflm21} introduce a privacy-preserving federated learning scheme with membership proof, to mitigate a special attack called \textit{deceiving attack} performed by a misbehaving server.
A deceiving attack aims to recover the gradients of the clients even when they are protected by secure aggregation (e.g., masking).  
To do so, the server deceives a subset of clients into thinking that the other clients are offline. 
To mitigate such deceiving attacks, Jiang et al. \cite{pflm21} suggest to prove the 
membership of clients whose gradients contributed to the global model by enforcing the server to announce the online users over the public ledger.  
Then, for verifiable aggregation, they leverage aggregated signatures, homomorphic encryption, and cryptographic accumulators. 

Buyukates et al. \cite{BuyukatesHHFZLFA23} also indirectly provide the functionality of the proof of participation.  
They propose a blockchain-empowered data marketplace that fairly rewards the clients according to their contributions to the model. Their approach logs each client's contribution in the public ledger and uses the zero-knowledge proofs for verifiability.  
Although the aforementioned works facilitate proof of participation, 
they do not guarantee either \textit{privacy} or \textit{anonymity} of the clients, as they are recorded in the public ledger.
Furthermore, because a curious party can analyze the public ledger to infer a client’s activity over time, \textit{unlinkability} is not achieved. 

\textit{\textbf{Watermarking-based methods.}}
Another approach is to embed a watermark during or after the generation of an FL model, either on the client side or the server side. 
In the client-side watermarking \cite{fedipr23,fedcip23,waffle21,flwm22,FedZKP}, clients insert a watermark to the model during the local training, so that they can prove their contribution by revealing their watermark on the model when they want. 
On the other hand, the server-side watermarking \cite{flareTekgulA23} (so-called \textit{fingerprinting}) focuses on detecting intellectual property infringement by tracking model leaks. 
While the aforementioned works allow only clients or the server to verify the watermark by accessing the watermark, only \cite{FedZKP} allows the clients to prove their participation (ownership) without revealing the watermark by using the zero-knowledge proofs. 
However, integrating watermarking techniques to FL settings faces four main challenges: 
1) (possible) accuracy loss of the model; 
2) undetectable watermarking after inserting a certain amount of noise; 
3) the necessity of accessing the model for watermark detection; and 4) the demand on a trusted third party in case of a dispute.
Moreover, privacy and anonymity of the clients are not considered as major features. 
During a proof of participation to a third party via watermark detection, the client identity is revealed. 
Thus, the third party can easily link the multiple participation of the same client. 
Therefore, the watermarking solutions do not provide unlinkability. 

In the light of the drawbacks above, we focus on a less explored approach facilitating cryptography and propose a solution requiring less computing effort for clients and service providers and without relying on a public ledger.

\section{Preliminaries}\label{preliminaries}
Let $\lambda \in \mathbb{N}$ be a security parameter. A probabilistic polynomial time (PPT) algorithm is a probabilistic algorithm taking $1^\lambda$ as input that has running time bounded by a polynomial in $\lambda$. A positive function $\mathsf{negl}: \mathbb{N} \rightarrow \mathbb{R}$ is called  \textit{negligible}, if for every positive polynomial $\poly(\cdot)$, there exists a constant $c > 0$ such that for all $x > c$, we have  $\mathsf{negl}(x) < 1/\poly(x)$. 
A hash function $\hash$ is a deterministic function from an arbitrary size input to a fixed size $param$ output, denoted $\hash:\{0,1\}^*\rightarrow \{0,1\}^{\lambda}$ where $\lambda \in \mathbb{N}$ is a security parameter. We deploy a collision resistant hash function. 
All the other primitives are discussed below and all other notation is shown in Table~\ref{tab:notation}.

\begin{table}[htb]
    \centering
    \caption{Notation.}
    \begin{tabular}{|c|l|}
        \hline
        \textbf{Notation} & \textbf{Interpretation} \\
        \hline
        $n$ & Number of clients \\
        \hline
        $n_{drop}$ & Maximum number of dropped-out clients \\
        \hline
        \clienti & Client $i$ \\
        \hline
        \clientset & All the clients in the system ($\{\clienti\}_{i\in[n]}$) \\
        \hline
        \posserver & Service Provider \\
        \hline
        \flserver & FL Server \\
        \hline
        \idata & Private data of \clienti \\
        \hline
        \model & Machine learning model \\
        \hline
        $\secretwitness_i$ & \clienti's secret witness \\
        \hline
        \globalpublicwitness & Global verification token \\
        \hline
        \proofstate & Proof of participation \\
        \hline
        \globalwitness & Group secret witness \\
        \hline
    \end{tabular}
    \label{tab:notation}
\end{table}

\subsection{Securely Aggregated Federated Learning} 
An FL system \cite{flMcMahanMRHA17} consists of multiple rounds. 
In each round $l$, the FL server \flserver sends the updated (global) model $\model$ to each client \clienti possessing private data $\idata$. 
\flserver selects a set of clients to contribute in each round.
\clienti computes its local model (updates) $\mlmsg$  over 
\idata (using the public parameter $\mlparam$) as $\Train(\mlparam,\idata)\rightarrow \mlmsg$ and later sends \mlmsg  to \flserver. 
Then, \flserver aggregates the received local models and transmits the updated model to the clients. 
This process is iterated until the model converges. 
Note that the aggregation procedure depends on the underlying structure. 
For instance, FedAvg \cite{flMcMahanMRHA17}  
averages the local updates during aggregation while FedQV \cite{fedqvtian24} aggregates using the voting strategy by exerting Quadratic Voting. 
\flserver, however, can still learn the information about \idata through inference and reconstruction attacks called \textit{privacy attacks}~\cite{FLvulner21}.\footnote{Although there are various attacks within FL, we only consider privacy attacks against clients, and client side attacks (e.g., poisoning) are considered out of scope.} 
To mitigate such attacks, a concept of \textit{secure aggregation (SA)} \cite{mcmahan2017communication} is proposed. 
The idea is to add a layer of security on \mlmsg before sending it to \flserver for aggregation. 
The SA protocol of \cite{bell2020secure}, which we implement of, consists of three phases: 
\begin{itemize}
\item 
$\SAGen(1^\lambda)\allowbreak
\rightarrow\allowbreak
(\clienti[\paramsa,\allowbreak\sarandb,\allowbreak\sakeyk,\allowbreak
\saclientspublic,\allowbreak
\{\saneighborpublic\}_{j\in \mathcal{N}(i)}],\allowbreak
\flserver[\paramsa,\allowbreak
\{\saclientspublic\}_{i\in[n]}])$: 
Each client \clienti generates: (1) a key pair (\sakeyk, \saclientspublic), (2) a random mask \sarandb, and (3) Shamir secret shares of both \sakeyk and \sarandb, denoted \sharekey and \sharerandom, shared with neighbors, denoted by $\mathcal{N}(i)$. 
The shares are end-to-end encrypted among clients with a secure authenticated encryption scheme and sent via \flserver which also collects and distributes public keys \saclientspublic among neighbors. 

\item 
$\SAProtect(\clienti[\paramsa,\allowbreak
\mlmsg,\allowbreak
\sarandb,\allowbreak
\sakeyk,\allowbreak
\saclientspublic,\allowbreak
\{\saneighborpublic\}_{j\in\mathcal{N}(i)}],\allowbreak
\flserver[\paramsa,\allowbreak
\{\saclientspublic\}_{i\in[n]}])\allowbreak
\rightarrow\allowbreak
\mlmsgprotected$: 
In the protect phase, \clienti masks its model update \mlmsg as follows: 
computes pairwise masks $\sapairwisemask := \pairwisemaskext$ for each client $C_j, \forall j \in \mathcal{N}(i)$ using Diffie-Hellman over its private key and neighbor's public key; constructs its masked update \mlmsgprotected as $
\mlmsgprotected := \mlmsg + \sarandb \pm \sum_{\substack{j \in \mathcal{N}(i) }} \sapairwisemask$;
and sends \mlmsgprotected to \flserver. 
\item 
$\SAAgg(\{\clienti[\mathsf{k}_{i},\allowbreak
\mathsf{b}_i,\allowbreak
\{\mathsf{s}^b_{i,j},\allowbreak
\mathsf{s}^k_{i,j}\}_{j \in \mathcal{N}(i)} )\}_{i\in[n]}],\allowbreak
\flserver[\{\saclientspublic,\allowbreak
\mlmsgprotected\}_{i\in[n]}])\allowbreak
\rightarrow\allowbreak
(\clienti[\model]_{i\in [n]},\allowbreak
\flserver[\model])$: 
Once \flserver identifies online clients \onlineclientset and dropouts \droppedclientset, it recovers and removes all masks by computing the aggregate. Then, it sends the aggregated result to online clients. 
Specifically, for each online neighbor $j \in \mathcal{N}(i)_o$, \clienti provides \sharerandom, while it provides \sharekey for each dropped-out neighbor $ j \in \mathcal{N}(i)_d$. 
Upon receiving \sharerandom and \sharekey,
\flserver reconstructs: 1) \sarandb 
using \sharerandom; and 2)
\saneighborkeyk and derives \sapairwisemask using \sharekey. 
Note that the server removes all masks, both self-masks and unmatched pairwise masks, and computes the final aggregated model update \model as
$
\sum_{i \in \onlineclientset} \mlmsg= \sum_{i \in \onlineclientset} \left( \mlmsgprotected - \sarandb \pm \sum_{\substack{j \in \droppedclientset}} \sapairwisemask \right). 
$ 
\flserver sends \model to $\{\clienti\} \in \onlineclientset$.
\end{itemize}
SA is assumed to be secure against the FL server/aggregator who attempts to learn individual gradients. 
Note that our approach is independent of the underlying FL SA (see Section \ref{discussion}). 

\subsection{Threshold Signatures}\label{threshold_signature}
A $(t,n)$-threshold signature scheme \cite{frost2020,lithresholdsig94} allows a group of $n$-many signers to jointly produce a single signature on the same message only if threshold-many ($t$) \textit{or more} of the signers participate in the signing process. 
The overall idea is to protect the secret signing key by splitting it into $n$ shares so that any $t$ or more members can generate a group signature using their shares. 
Later, a verifier can validate the group signature without identifying the signers. 
A threshold signature scheme consists of the following four main algorithms: 
\begin{itemize}
    \item $\TSKeyGen(\tsparam,n,t)\rightarrow(\allpk,\{\pki,\ski\}_{i\in n})$: On inputs of the public parameter \tsparam, the number of parties $n$, and a threshold $t$, \TSKeyGen generates a signing key \ski and a verification key \pki for each signer, and a global verification key \allpk. 
    \item $\TSSign(\msg,\ski)\rightarrow \sig_i$: The partial signing algorithm allows a signer to generate a signature $\sig_i$ on a pre-agreed message \msg using its signing key \ski. 
    \item $\TSSignAgg(\{\sig_i\}_{i\in[t]},\allpk)\rightarrow \sig$ : Upon receiving at least threshold $t$-many signatures generated by the signers, it returns the aggregated signature $\sig$. 
    \item $\TSVerify(\sig,\msg,\allpk) \rightarrow \top/\perp$: It outputs the verification result of a given signature \sig, for the given message \msg and the global verification key \allpk. 
\end{itemize}
A threshold signature scheme is assumed to be secure against existential unforgeability under chosen-message attacks (EUF-CMA) by showing that it is difficult for an adversary to forge signatures even when performing an adaptive chosen-message attack. Furthermore, a threshold signature scheme ensures signer indistinguishability that hides the identities of the signers.

\subsection{Oblivious Pseudorandom Function}\label{oprf}
An oblivious pseudorandom function (OPRF) \cite{oprfsok22} is a protocol between two parties, \textit{sender} and \textit{receiver}, that securely compute a pseudorandom function $\mathsf{F}_\prfkey(x)$, where $\prfkey$ is the sender's input and $x$ is the receiver's input. 
The sender learns \textit{nothing} from the interaction, and the receiver learns \textit{only} $\mathsf{F}_\prfkey(x)$. 
For our construction, \texttt{2DashDH}~\cite{oprfsok22} is deployed which is defined as $\mathsf{F}_\prfkey(x):=\hashtwo(x,(\hashone(x))^\prfkey)$ for $\prfkey \leftarrow \Zq$, where $\hashone$ and $\hashtwo$ are hash functions that map arbitrary-length strings into elements of $\{0, 1\}^\lambda$ and $\mathbb{G}$, respectively. 

\section{System Overview and Security Requirements}\label{system_overview}

\subsection{System Model}\label{architecture}
\FedPoP is inspired by the existing FL designs with intellectual property protection \cite{fedright23} and operates among: 
\textit{n-many} \textit{clients}, an FL \textit{server}, and a \textit{service provider}. 
Clients, denoted by $\clientset=\{\clienti\}_{i\in[n]}$, are the contributors of an FL service with their private data \idata. 
They compute a model generation with an FL server and receive a global model and a proof for their participation to the model as rightful owners.\footnote{Intellectual protection of their data such as \cite{freqywm22} is considered out of our scope.} 
An FL server, denoted by \flserver, is a server willing to train a model \model with the clients' private data. 
\flserver is the entity who determines the model, e.g., model parameters. 
A service provider, denoted by \posserver, possesses a model $\model'$ to provide service to its users. 
A client \clienti wishes to prove its contribution to $\model'$ held by \posserver, which it assumes corresponds to a model it previously contributed to.

\subsection{Threat Model}\label{securityproperties} 
\FedPoP is designed to be secure against a semi-honest \flserver and a semi-honest \posserver, while secure against \emph{at most} threshold-many malicious clients. 

\noindent \textbf{Malicious Clients.} \textit{At most} a threshold many malicious clients attempt either to impersonate another client or to convince an honest service provider in order to falsely claim their participation. 

\noindent \textbf{Honest-but-curious FL Server.} While \flserver faithfully follows the protocol, it attempts to deduce information about clients' private local models.

\noindent \textbf{Honest-but-curious Service Provider.} \posserver attempts to accomplish any one of the following goals: 1) identifying the client who is claiming the proof of participation; 2) linking multiple generations of models to the same client; and 3) identifying other clients involved in the generation of the model.

We assume that clients do not poison the local models, i.e., poisoning attacks \cite{poisonattack20}, and consider it as out of scope. 
For communication, we assume secure and authenticated channels between parties (i.e. between \clientset and \flserver, and between \clienti and \posserver), which can be performed using standard means such as HTTPS. 

\subsection{Security Requirements (SR)}\label{securityprivacyreq} 
Based on our threat model, we devise four security requirements. 
We briefly introduce them below while we provide their formal security game definitions in the full (anonymous) version \cite{fedpopfullversion} due to page limitation. 
In the next section, we introduce the ideal functionality definition of \FedPoP. 

\noindent\textbf{SR1. Soundness.} 
Even if at most threshold many malicious clients collude, they should not be able to 
forge a proof or convince an honest service provider to accept their (false) claim of participation. \\
\textbf{SR2. Privacy.} While \clienti proves its participation to a model, the proof shall not reveal the identities of the other clients to \posserver.\\
\textbf{SR3. Anonymity.} The identity of \clienti proving its participation shall be kept private from \posserver. 
Note that we do not consider the identification of \clienti from its internet protocol (IP) address; however, this can be prevented via anonymous communication channels such as Tor. \\ 
\textbf{SR4. Unlinkability.} An honest-but-curious \posserver shall not link multiple proofs of participation to the same client \clienti. 

\subsection{Intuition}\label{motivation} 
At a high level, our overarching goal is to devise a privacy-preserving federated learning system that supports verifiable proof of participation that is used as proof of ownership. A natural starting point is to use digital signatures atop a securely aggregated FL protocol. In this basic approach, the \flserver signs the public verification key of each participating \clienti using its private signing key. Each \clienti, equipped with its own signing and verification keys, receives a signature alongside the trained model \model. To prove participation, \clienti later sends this signature and its verification key to \posserver, and engages in a challenge-response protocol to demonstrate knowledge of the corresponding private (signing) key. 
This method, however, has key limitations. It reveals the identity of \clienti, violating anonymity; it enables linkability across multiple proofs; and it relies on a central signing authority, introducing a single point of failure. These issues motivate our construction, \FedPoP, which we develop by addressing the following challenges.
\begin{itemize}
\item \textbf{Challenge 1 [Joint Proof of Participation].}
To resist forgery and have a more robust solution, clients should collectively generate a participation proof after receiving the model. We adopt threshold signatures, which allow a subset of contributing clients to jointly sign using their individual private keys. The resulting signature is verifiable using a global verification key derived from their public keys. Unlike ring or multi-signatures, threshold signatures satisfy our functional and security goals. A detailed comparison is provided in Appendix~\ref{otherinstances}.
\item \textbf{Challenge 2 [Possession of the Private Witness].}
While threshold signatures confirm joint participation, they do not by themselves prove that an individual client knows its private key. One option is to use zero-knowledge proofs (ZKPs) to demonstrate both contribution to the global verification key and knowledge of the corresponding secret. However, ZKPs often incur high costs and struggle to preserve unlinkability, especially for large models. 
Another option is to use cryptographic accumulators, as in \cite{pflm21}, to prove inclusion in a global key set. But this approach still requires proving ownership of the private key, and exposes the list of client keys, violating anonymity (\textbf{SR2}), unlinkability (\textbf{SR3}), and privacy (\textbf{SR4}). To overcome these limitations, we introduce a fresh \textit{group witness} during joint proof of participation generation. This group witness, which is realized as a pseudorandom key, is issued only to the participating clients. It enables the creation of an auxiliary proof component using the global verification key. 
Later, when a client \clienti wishes to prove its participation, \clienti and \posserver perform an Oblivious PRF (OPRF) protocol on the key (held by \textit{only} \clienti) and group verification key (held by \posserver). By comparing the output with a precomputed value, \posserver is convinced of the client’s possession of the group witness without learning it. 
\end{itemize}

Hence, \FedPoP design integrates threshold signatures with OPRF to yield a novel, efficient, and privacy-preserving proof mechanism in federated learning. The global verification key remains unlinkable to clients, and the proof of secret (group) witness possession avoids costly zero-knowledge constructions.

\section{\FedPoP Design}
This section begins with the \FedPoP architecture followed by its formal security definition. Later, we introduce our \FedPoP instance satisfying all the aforementioned requirements.

\subsection{\FedPoP Architecture} 

We now introduce the design of \FedPoP in details. 
As depicted in Fig. \ref{fig:architecture}, \FedPoP comprises three phases: \Setup, \Generate, and \Prove. 

\begin{figure*}[ht!]
  \centering
  \includegraphics[scale=0.55]{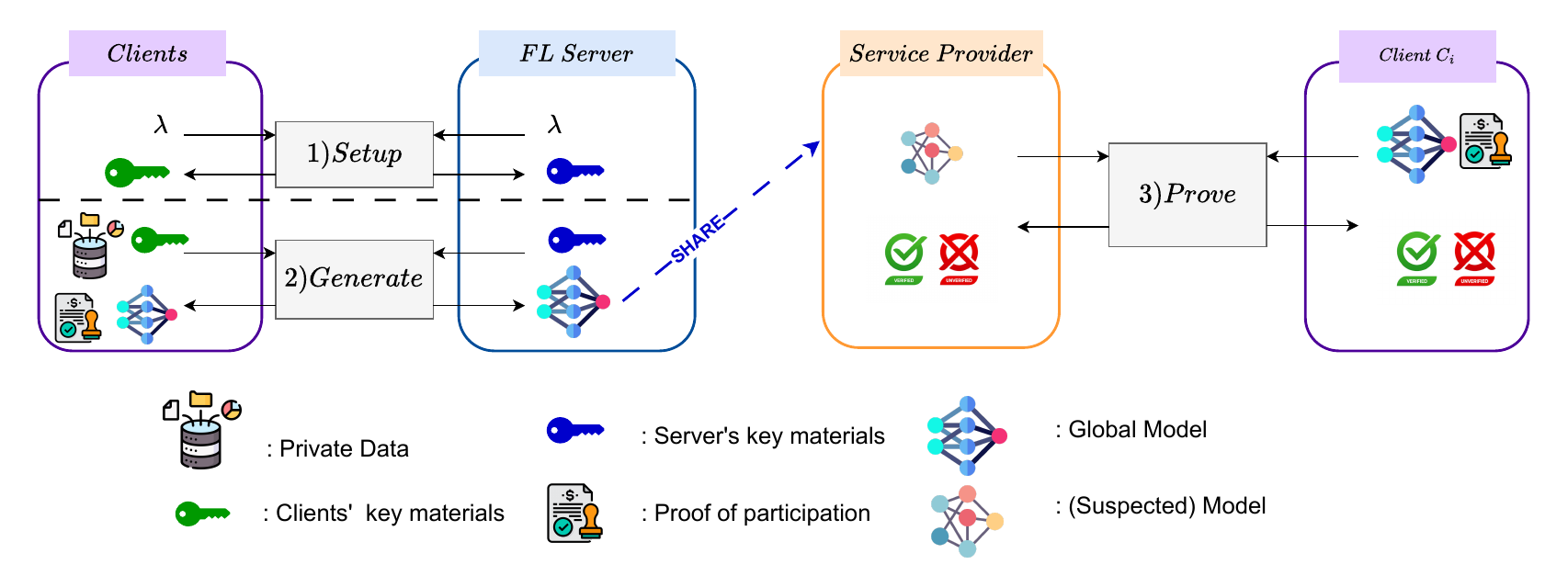}
    \caption{\FedPoP Architecture. }
    \label{fig:architecture}
\end{figure*}

\textbf{Setup Phase}. This phase is executed between \clientset (a set of $n$ many clients) and \flserver. \clientset and \flserver compute and obtain their corresponding key materials, \clientkeymaterials and \serverkeymaterials, respectively. A \clientkeymaterials consists of an individual secret witness $\secretwitness_i$ (which is kept secret by clients), a verification token $\publicwitness_i$, aggregation keys, the model parameters, and public parameters of the underlying primitives. \serverkeymaterials consists of the system public parameters, keys generated for secure aggregation, and the model parameters. 
The setup phase is shown as : $\Setup(\param)\rightarrow\clienti[\clientkeymaterials]_{i\in[n]},\flserver[\serverkeymaterials]$.

\textbf{Generate Phase}. In this phase, \clientset and \flserver engage to securely train the model on the clients' private data \idata using their corresponding key materials from \Setup. 
After the \Generate protocol execution, each (active) \clienti receives the trained model \model, a group private witness \globalwitness, a global verification token \globalpublicwitness, and a proof of participation \proofstate while \flserver receives \model and \globalpublicwitness. 
We assume that \globalpublicwitness is public information known by \textit{all the parties} in the system. 
The generate phase is shown as
$
\Generate(\{\clienti[\idata,\clientkeymaterials]\}_{i\in[n]}, \flserver[\serverkeymaterials]) \rightarrow \{\clienti[\model,\globalwitness, \globalpublicwitness, \proofstate ]\}_{i\in[n]}, \flserver[\model, \globalpublicwitness].
$

\textbf{Prove Phase}. This phase takes place between \clienti and \posserver. 
\clienti wishes to prove to \posserver that it participated in the generation of a (suspected) model possessed by \posserver, denoted by $\model'$. 
\clienti, holding $\langle \model,\globalwitness,\globalpublicwitness,\proofstate\rangle$, and \posserver, holding $\langle \model',\globalpublicwitness'\rangle$, jointly compute the \Prove phase. 
At the end of the protocol, \posserver returns $1$ if it is convinced by the proof given by \clienti, without knowing its identity, and \clienti and \posserver without knowing each other's participating models, i.e., $\model$ and $\model'$. 
Otherwise, it returns $0$. 
The phase is shown as $\Prove(\clienti[\model,\proofstate, \globalwitness],\posserver[\model',\globalpublicwitness]) \rightarrow \clienti[0/1], \posserver[0/1]$.

\subsection{Security Definition}\label{sec:securitydefinition} 
We formally define the problem of proof of participation as overviewed previously via the ideal-real world paradigm \cite{canetti2000}. 

\textbf{Ideal World.} An ideal functionality is defined to describe the security properties expected by a system. 
The ideal functionality is an ideal trusted party that performs the task expected by the system in a trustworthy manner. 
When devising an ideal functionality, one describes the ideal properties that the system should achieve, as well as the information that the system will inherently leak.
\clientset, \flserver, and \posserver with their inputs wish to compute \FedPoP by engaging with the ideal functionality, denoted by \Fpop as formally defined below.

\begin{functionality}[label=pomfunc]{\Fpop}
Parameterized with a set of clients $\mathcal{C}=\{\clienti\}_{i\in [n]}$, a federated learning server \flserver, and a service provider server \posserver, \Fpop stores all the messages and works as follows, where $\mlparam$ is a (initial) model parameter determined by \flserver and \idata is the data belonging to \clienti: 
\begin{itemize}
\item  Upon receiving $(\Setup,\{\IDi\}_{i\in [n]})$ from each $\clienti \in \mathcal{C}$,
\begin{itemize}
    \item If \IDi exists, abort and return $\langle \abort \rangle$ to \clienti. 
\end{itemize}
Otherwise:
\begin{itemize}
    \item Generate an individual secret witness and verification token $( \secretwitness_i, \publicwitness_i)$ and send $\initok$ to \clienti. 
\end{itemize}
\item Upon receiving $\langle \Generate,l,\{\idata\}_{i\in [n]}\rangle$ from all the clients for a given state (time/round) $l$ and $\langle \Generate,l,\mlparam\rangle$ from \flserver,
 \begin{itemize}
 \item Compute a global model $\model^l$ based on $\{\idata\}_{i\in [n]}$ and \mlparam, a proof $\proofstate^l$, a group secret witness \globalwitness, and a global verification token $\globalpublicwitness^l$.
    \item Send $\langle \model^l,\proofstate^l,\globalwitness,\globalpublicwitness^l\rangle $ to each \clienti and \flserver.  
 \end{itemize}
\item Upon receiving  (\Reveal, $l$) from \flserver for a given state $l$, send the model and its corresponding global verification token $\langle \model^l,\globalpublicwitness^l\rangle$ to \posserver.
 
\item Upon receiving  $\langle \Prove,\model^l,\proofstate^l, \globalwitness\rangle$ from a client \clienti and $\langle \Prove,\model',\globalpublicwitness\rangle$ from \posserver:
\begin{itemize}
    \item Send $1$ to \clienti and \posserver if $\model^l = \model'$ and $(\clienti,\model^l,\proofstate^l,\globalpublicwitness^l,\globalwitness)$ is in the record; otherwise, send $0$ to \clienti and \posserver. 
\end{itemize}
\end{itemize}
\end{functionality}

\textbf{Real World.} The real world consists of \clientset, \flserver, and \posserver. 
The parties are involved in the real world execution of a \FedPoP instance, denoted by \Ppop. 
Note that there is no universal trusted entity as \Fpop for a real world \FedPoP protocol. 

\begin{definition}\label{def:pomsecdef}
Let \Ppop be a PPT protocol for federated learning with proof of participation. 
We say that \Ppop is secure if for every non-uniform PPT real-world adversary $\Adv$, there exists a non-uniform PPT ideal world simulator \Dist such that all transcripts among parties and outputs between the real and ideal worlds are computationally indistinguishable;

$$\Ideal_{\Fpop}^{\Dist(aux)}(\{\idata,\secretwitness_i,\publicwitness_i\}_{i\in[n]},\model,\model', \proofstate,\globalwitness,
\globalpublicwitness,\lambda)  \nonumber
	   \approx_c  \nonumber $$ $$
	 \Real_{\Ppop}^{\Adv(aux)}(\{\idata,\secretwitness_i,\publicwitness_i\}_{i\in[n]},\model,\model', \proofstate,\globalwitness, \globalpublicwitness,\lambda) \nonumber $$
	\label{def:pom} 
\noindent where $aux \in \{0,1\}^*$ is the auxiliary input and $\lambda$ is the security parameter.
\end{definition}

\section{Our Proposed \FedPoP Instance}\label{fedpopinstance}
We now introduce our \FedPoP construction leveraging well-known cryptographic primitives considering the aforementioned key features and requirements. 
We assume that \Setup and \Generate phases executed per round $l$ as below. 
In Section \ref{discussion}, we discuss multi-round and possible approaches to increase the efficiency. 

\noindent \textbf{Setup Phase}. 
During \Setup phase, the following key generation algorithms are computed. 
For the sake of concreteness, we assume that setup algorithms are instantiated securely. 
This can be achieved with a suitable distributed key generation (DKG) protocol~\cite{kate2009distributed} (e.g., the DKG protocol \cite{dkgCritesKM21} for threshold signatures) or simply with a trusted dealer. 
Our approach is independent of the specific key generation methods, as long as the resulting keys satisfy some basic correctness and security conditions. 
For simplicity, we assume a trusted dealer. 
However, solutions such as FLAMINGO \cite{ma2023flamingo} could be deployed for SA in order to eliminate the necessity of per round setup. \\ 
\noindent $1)$ \textit{\textbf{Threshold signature key generation}}: $\TSKeyGen(\tsparam,n,t) \rightarrow(\allpk,\{\pki,$ $\ski\}_{i\in [n]})$ where each \clienti receives the global verification key \allpk and their own verification and secret signing keys $\langle\pki,\ski\rangle$ while \flserver receives only \allpk. 
Threshold $t$ is defined as $n - n_{drop}$ where $n_{drop}$ represents the maximum number of dropped-out clients the system can tolerate; thus, $t$ can be determined by \flserver. \\
\noindent $2)$ \textit{\textbf{Secure aggregation setup}}: \clientset and \flserver compute the setup as $\SAGen(1^\lambda)$ $\rightarrow (\clienti[\paramsa,\sarandb,\sakeyk,\saclientspublic,\{\saneighborpublic\}_{j\in \mathcal{N}(i)}], \flserver[\paramsa,\{\saclientspublic\}_{i\in[n]}]$ where \clienti receives $\langle \paramsa,\sarandb,\sakeyk,$ $\saclientspublic,\{\saneighborpublic\}_{j\in \mathcal{N}(i)} \rangle$ and the public parameter \paramsa while \flserver receives $\langle \paramsa,\{\saclientspublic\}_{i\in[n]}\rangle$. 
As a result, each \clienti receives $\langle \paramsa, \sarandb,\sakeyk,\saclientspublic,\{\saneighborpublic\}_{j\in \mathcal{N}(i)},$ $\pki,\ski,\allpk\rangle$ as its \clientkeymaterials and stores them securely. 
\flserver receives $\langle \paramsa,\{\saclientspublic\}_{i\in\clientset},\allpk \rangle$ as its \serverkeymaterials.

\begin{figure}[ht!]
\hspace*{-2.3cm}
  \includegraphics[scale=0.2]{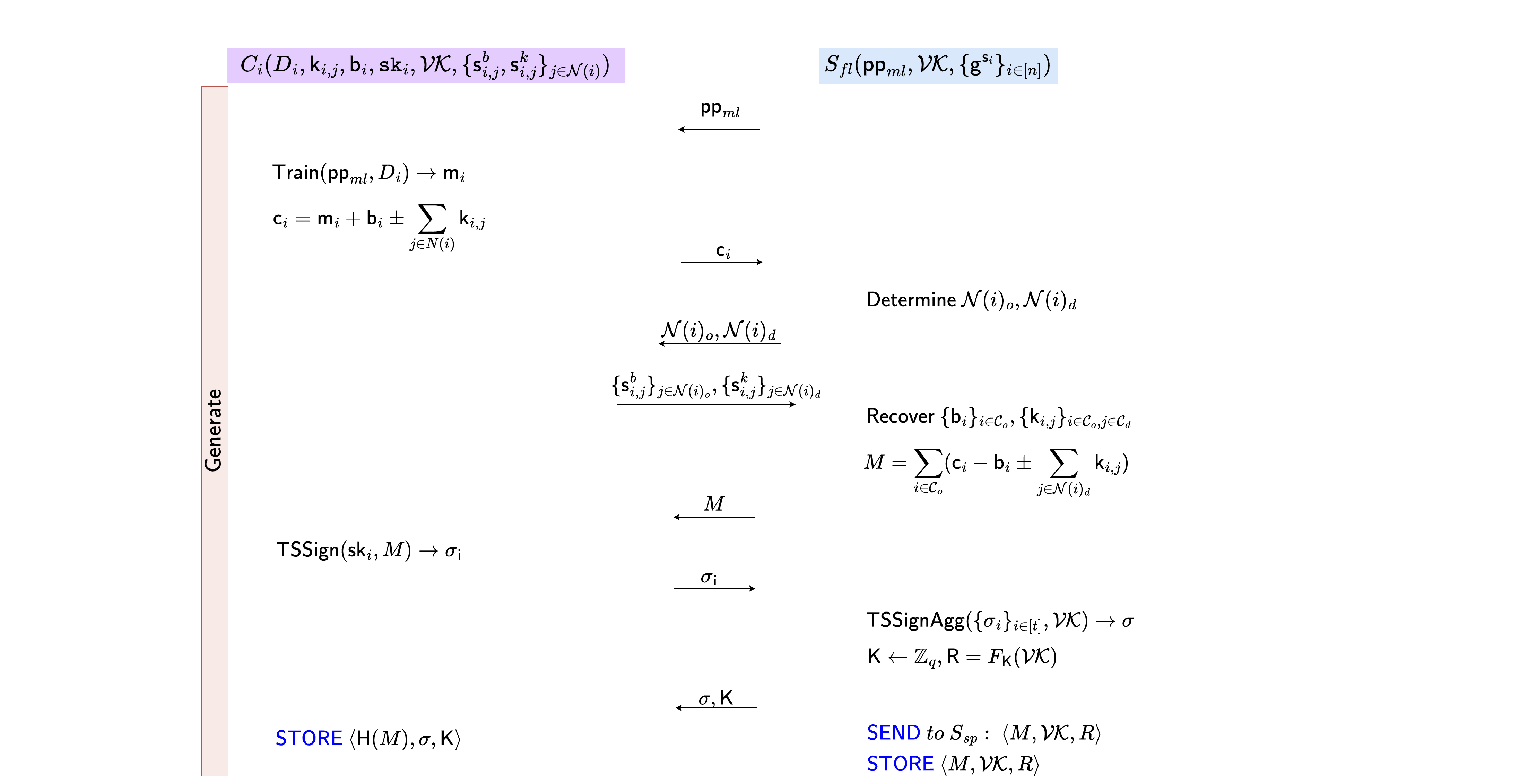}%
    \caption{ \Generate Phase.}
    \label{fig:fedpopGenerate}
\end{figure}
\noindent \textbf{Generate Phase}. 
In \Generate phase as illustrated in Fig. \ref{fig:fedpopGenerate}, \flserver sends model parameters \mlparam to each client \clienti. 
Upon receiving \mlparam, the clients generate their local models \mlmsg using their private data \idata as $\Train(\mlparam,\idata)\rightarrow \mlmsg$.  

Each client \clienti then blinds its local model \mlmsg by executing the \SAProtect function using a random self-mask \sarandb and pairwise masks \sapairwisemask for each neighbor $j$ ($j \in \mathcal{N}(i)$), resulting in a protected local model \mlmsgprotected:
$
\mlmsgprotected = \mlmsg + \sarandb \pm \sum_{\substack{j \in \mathcal{N}(i) }} \sapairwisemask$. 
\clienti sends \mlmsgprotected to \flserver. 
After a timeout period determined by \flserver, the server classifies clients whose \mlmsgprotected messages have been received as online clients, denoted by \onlineclientset, and labels the rest as dropped clients, denoted by \droppedclientset.
\flserver then determines the sets of online and dropped neighbors for each client, $\mathcal{N}(i)_o$ and $\mathcal{N}(i)_d$, and shares this information with each client \clienti. 
Upon receiving $\mathcal{N}(i)_d$ and $\mathcal{N}(i)_o$, for each online neighbor $j \in \mathcal{N}(i)_o$, \clienti provides \sharerandom to allow \flserver to recover and remove the self-mask \sarandb and \sharekey for each dropped out neighbor $ j \in \mathcal{N}(i)_d$. 
\flserver receives the corresponding shares and reconstructs individual masks \sarandb for \onlineclientset and pairwise masks \sapairwisemask for \droppedclientset. 
The server then removes all masks, both self-masks and unmatched pairwise masks, and computes the final aggregated model update \model as
$
\model = \sum_{i \in \onlineclientset} \mlmsg= \sum_{i \in \onlineclientset} \left( \mlmsgprotected - \sarandb \pm \sum_{\substack{j \in \droppedclientset}} \sapairwisemask \right). 
$ 
\flserver then sends the aggregated \model back to \onlineclientset. 
Upon receiving \model, each \clienti generates a partial signature $\sig_i$ on \model as $\TSSign(\model,\ski)\rightarrow \sig_i$. 
The clients send $\langle \sig_i\rangle$ to \flserver. 
Upon receiving at least threshold $t$-many partial signatures $\{\sig_i\}_{i\in[t]}$, \flserver aggregates $\{\sig_i\}_{i\in[t]}$ to produce the aggregated signature $\sig$ as $\TSSignAgg(\{\sig_i\}_{i\in[t]},\allpk)\rightarrow \sig$. 
Additionally, \flserver generates a PRF key as $\prfkey \leftarrow \Zq$ which is a group witness \globalwitness enabling only participated clients to prove their participation in the \Prove phase. \footnote{\prfkey can be also generated using \secretwitness's of \onlineclientset as detailed in Appendix \ref{sec:prfgeneration_alternative}.}
\flserver sends $\langle\sig,\prfkey\rangle$ to the clients and shares $\langle \model,\allpk,\prfvalue \rangle $ with \posserver where $\prfvalue=F_{\prfkey}(\allpk)$. 
\flserver stores the tuple of $\langle \model,\allpk,\prfvalue \rangle $ in which $\globalpublicwitness=\{\allpk,\prfvalue\}$. 
\clientset stores $\langle \model, \prfkey, \sigma\rangle$ where $\proofstate=\{\sigma\}$ and 
$\globalwitness=\prfkey$. 
However, note that the model storage could grow unexpectedly due to the multiple FL contributions of the clients. 
Therefore, alternatively \clientset can store the hash value of \model rather than \model itself together with \prfkey and \sig.

\begin{figure}[htb]
\hspace*{-0.4cm}
\includegraphics[
    scale=0.3,
    trim={7cm 0 5cm 0},  
    clip
  ]{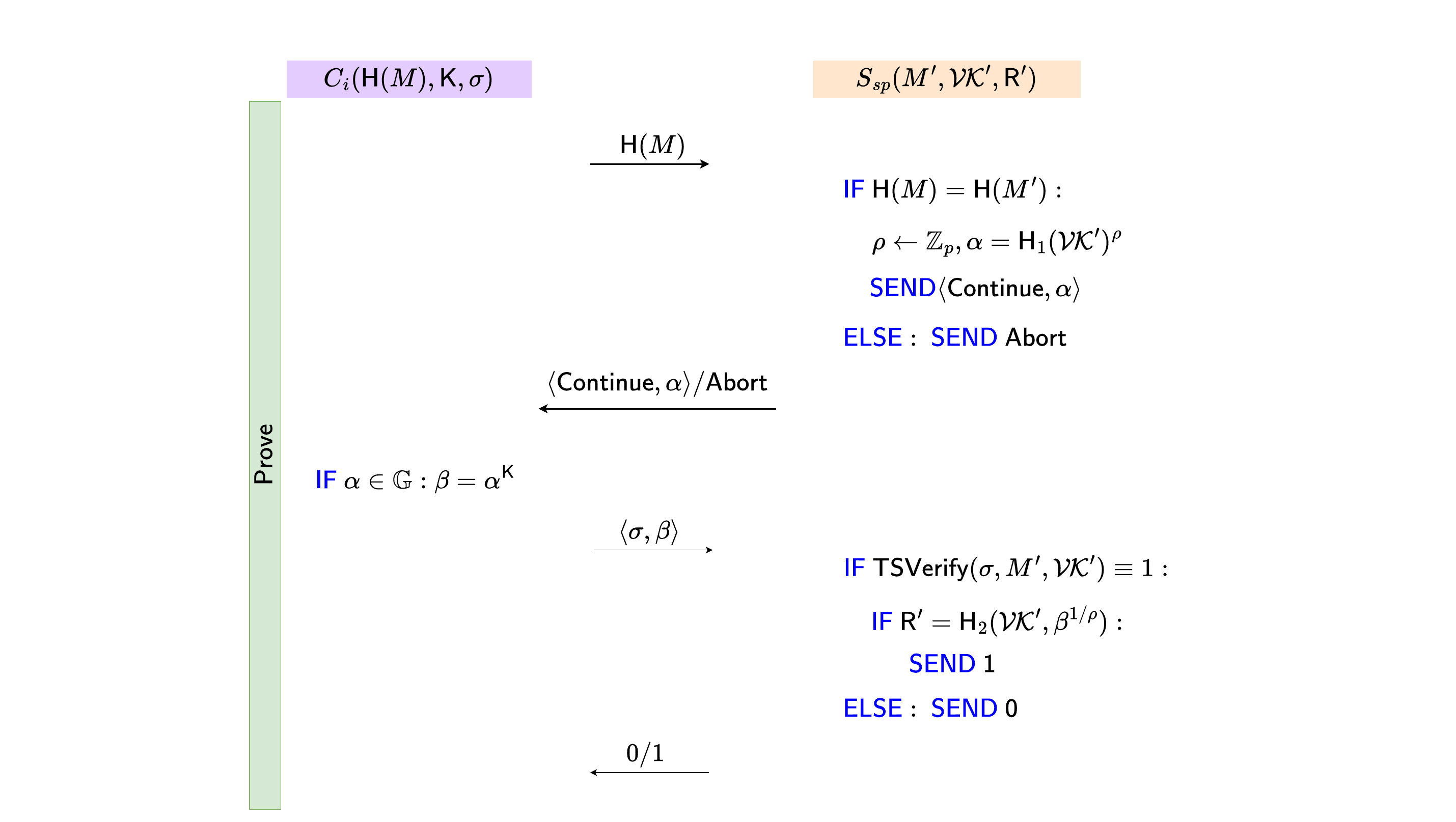}
  \caption{\Prove Phase.}
  \label{fig:fedpopProve}
\end{figure}

\noindent \textbf{Prove Phase.} 
\Prove phase as illustrated in Fig. \ref{fig:fedpopProve} is taken place between \clienti and \posserver. 
\clienti's inputs are the hash value of \model (or \model itself), the signature $\sig$ as \proofstate, and the PRF key \prfkey as \globalwitness. 
\posserver's inputs are a model $\model'$, a global verification key $\allpk'$, and the PRF value $\prfvalue'$ (which is $\mathsf{F}_\prfkey(\allpk')$). 
\clienti sends $\hash(\model)$ to \posserver. 
To check if \model is indeed possessed by \posserver, \posserver compares $\hash(\model)$ with $\hash(\model')$. 
If this check is successful meaning that \posserver possesses \model of \clienti, \clienti and \posserver compute the OPRF protocol where \clienti participates  with its input $\prfkey$ acting as a sender and \posserver participates with its input $\allpk'$ acting as a receiver. 
Hence, \posserver picks $\rho \leftarrow \Zp$ and computes $\alpha = \hashone(\allpk')^\rho$. 
\posserver sends $\alpha$ to \clienti. 
If $\alpha \in \G$, \clienti computes $\beta=\alpha^\prfkey$ and sends $\langle \sig, \beta\rangle$ to \posserver.  
\posserver first checks if \sig verifies under $\allpk'$ by computing  $\TSVerify (\sig,\model',\allpk')$. 
If the signature verification is successful, \posserver computes $\hashtwo(\allpk',\beta^{\frac{1}{\rho}})$ and checks if the result equals to $\prfvalue'$. 
If this final check is also successful, \posserver outputs $1$ which indicates that \clienti indeed participated to the generation of $\model'$. 
It outputs $0$, otherwise. 

\section{Analysis}\label{sec:analysis_section}
\subsection{Security Analysis}\label{security}
This section discusses how \FedPoP satisfies the security requirements described in Section \ref{securityprivacyreq}. 
A probabilistic polynomial time (PPT) adversary \Adv corrupts a client while \posserver and \flserver are semi-honest. \Adv cannot corrupt all the parties at the same time, e.g., \Adv cannot corrupt \posserver and \flserver or \clienti and \posserver, simultaneously. 
The security analysis adopts the ideal/real-world methodology. In the \emph{ideal world}, a simulator $\Sim$ interacts with the functionality $\Fpop$ and in the \emph{real world}, it interacts with the actual protocol $\Pi$ and an adversary \Adv corrupting $\posserver$ or $t-1$ clients. The goal is to show that, for every PPT adversary, the view produced in the real world is computationally indistinguishable from the view produced in the ideal world. The argument proceeds via a short sequence of hybrid games ($H$). 
Since honest-but-curious \flserver runs a secure aggregation with clients, and partial signatures \flserver receives from the clients does not have information about individual $\mlmsg$, \FedPoP is secure against honest-but-curious \flserver. Therefore, we omit its formal proof. 

\begin{theorem}\label{thm:soundnessproof}
Our \FedPoP is secure according to Definition \ref{def:pomsecdef} against any non-uniform PPT adversary \Adv corrupting $t-1$ clients denoted by $\clientset_c$ assuming that the threshold signature is secure, the oblivious pseudorandom function is secure, and
the hash function is collision-resistant.
\end{theorem}
\begin{proof} 
The simulator \Sim simulates honest parties which are the \flserver, \posserver, and $(n-t+1)$-many clients, denoted by $\clientset_h$ in the real world and malicious parties in the ideal world which are $t-1$ many clients denoted by $\clientset_c$. Note that \Sim stores all the data it receives, generates, and sends in its database. \Sim behaves as follows: 

\noindent \textbf{Setup Phase:}
\begin{itemize}
    \item \Sim Sends $(\Setup,\IDi)$ to $\Fpop$. If the response is $\langle \abort \rangle$, forward it to $\Adv$. Otherwise, it receives $\initok$ from $\Fpop$.
    \item \Sim together with \Adv computes the threshold signature key generation algorithm resulting $(\allpk', \{\pki',\ski'\}) \leftarrow \TSKeyGen$ where \Adv receives $(\allpk', \{\tilde{\pki}',\tilde{\ski}'\}_{i \in \clientset_c})$ and \Sim $(\allpk', \{\pki',\ski'\}_{i\in \clientset_h})$,$\paramsa, \sarandb,\sakeyk,\{\sapairwisemask\}_{j\in \mathcal{N}(i)}, \prfkey \leftarrow \mathbb{Z}_q$, 
    $\langle \paramsa,\sarandb,\sakeyk,$ $\saclientspublic,\{\saneighborpublic\}_{j\in \mathcal{N}(i)} \rangle$ and the public parameter \paramsa while \flserver receives $\langle \paramsa,\{\saclientspublic\}_{i\in[n]}\rangle$.  
    \item \Sim stores all the values in its database.
\end{itemize}
\textbf{Generate Phase:}
\begin{itemize} 
\item \Sim sends $\{\idata'\}_{i\in\clientset_h}$ to \Fpop.
 \item Upon receiving $(\model,\proofstate=\sig,\globalwitness = \prfkey,\globalpublicwitness=$ $\{\allpk,\prfvalue\})$ from \Fpop, \Sim sends $\mlparam'$ to \Adv in the real world and computes $\Train(\mlparam',\idata')\rightarrow \mlmsg'$ and for each neighbor $j$ ($j \in \mathcal{N}(i)$), resulting in a protected local model \mlmsgprotected: $\mlmsgprotected' = \mlmsg' + \sarandb \pm \sum_{\substack{j \in \mathcal{N}(i) }} \sapairwisemask$. 
  \item Upon receiving $\{\tilde{\mlmsgprotected}\}_c$ from \Adv ($\clientset_c$), \Sim sends $\mathcal{N}(i)_d$ and $\mathcal{N}(i)_o$ to \Adv. \Adv provides \sharerandom to allow \Sim to recover and remove the self-mask \sarandb and \sharekey for each dropped out neighbor $ j \in \mathcal{N}(i)_d$. \Sim receives the corresponding shares and reconstructs individual masks \sarandb for \onlineclientset and pairwise masks \sapairwisemask for \droppedclientset. \Sim computes the final aggregated model update \model as $\model' = \sum_{i \in \onlineclientset} \mlmsg= \sum_{i \in \onlineclientset} \left( \mlmsgprotected - \sarandb \pm \sum_{\substack{j \in \droppedclientset}} \sapairwisemask \right)$. 
    \item \Sim sends $\model'$ to \Adv. 
    \item Upon receiving $\{\tilde{\sig_i}\}_c$ from \Adv, \Sim generates a signature for the rest $\TSSign(\model',\ski')\rightarrow \sig_i'$ and then $\TSSignAgg(\{\sig_i'\}_h \cup \{\tilde{\sig_i}\}_c,$ $\allpk')\rightarrow \sig'$. Finally, \Sim generates a OPRF key $\prfkey' \leftarrow \Zq$ from the same PRF key distribution.
    \item \Sim sends $\langle \sig', \prfkey' \rangle$ to \Adv. 
\end{itemize}
\textbf{Prove Phase:}
Upon receiving $\tilde{h}$ from \Adv, \Sim checks if $\tilde{h}= \hash(\model')$ and proceed with one of the following cases depending on the result of the comparison: 
\begin{itemize}
    \item \textbf{Case 1.} $\tilde{h}\neq \mathsf{H}(\model')$: \Sim sends $\langle M',\proofstate=\{\sig\}\rangle$ to \Fpop and receives $0$ the result from \Fpop and forwards \texttt{Abort} ($0$) to \Adv.
    \item \textbf{Case 2.} $\tilde{h}= \mathsf{H}(\model')$: \Sim generates $\rho' \leftarrow \Zp$ and computes $\alpha' = \hashone(\allpk)^{\rho'}$.
    \begin{itemize}
        \item \Sim sends $\alpha'$ to \Adv. 
        \item \Sim receives $\langle \tilde{\sig},\tilde{\beta}\rangle$ from \Adv.
        \item There are two sub cases to consider after receiving $\langle \tilde{\sig},\tilde{\beta}\rangle$: 
        \begin{itemize}
            \item \textbf{Case 2.1}: If $\tilde{\sig} \in DB$ and $F_{\prfkey}(\allpk)=\hashtwo(\allpk,\tilde{\beta}^{\frac{1}{\rho'}})$, then 
            \Sim sends $\langle M,\proofstate=\{\sig\}\rangle$ to \FedPoP. 
            \item \textbf{Case 2.2}: Otherwise, \Sim sends $\langle M',\sig'\rangle$ to \FedPoP.
        \end{itemize}
   Whichever the above cases occur, \Sim forwards the result sent by \FedPoP to \Adv.
    \end{itemize}
\end{itemize}
\end{proof}

\begin{claim}
    The view of \Adv, corrupting $t-1$ many clients, in their interaction with \Sim is indistinguishable from the view in their interaction with the real-world honest parties. 
\end{claim}
We argue via a standard sequence of hybrid games $G_i$. $G_0$ is the real execution of \FedPoP. 
In each subsequent hybrid we replace exactly one step of the real protocol with its simulated counterpart, changing only the adversary's view. 
The final hybrid $G_i$ is the execution of the simulator \Sim defined above, which operates without the honest parties' actual inputs. 
For every adjacent pair ($G_{i-1}, G_i$) we provide a reduction showing computational indistinguishability under the stated assumptions. 
We now detail each hybrid and its corresponding reduction: 

\begin{itemize}
  \item \textbf{Game $G_0$:} In this game, we use the inputs of honest parties as the inputs of our simulation. This game is identical to real-world \FedPoP protocol. 
  \item \textbf{Game $G_1$:} This game is the same as $G_0$ except when running the SA protocol, \Sim used a different \mlparam, \idata, $\sarandb$, and $\sapairwisemask$ which are respectfully derived from their space from the ideal world, and it resulted in $\model'$. 
  \Adv can distinguish this behavior from $G_0$ only if it can break the secure aggregation confidentiality or the collusion resistant hash function. However, secure aggregation security hides the true values using random key values and prevents \Adv from learning $\idata'$ which is used instead of true data while hash function ensures that \Adv cannot cheat while sending $\tilde{h}$ since the only way for \Adv to succeed is to break the security of hash function. Hence, $G_1$ is computationally indistinguishable from $G_0$. 
  \item \textbf{Game $G_2$:} This game is the same as $G_1$ except this time the real signing and verification keys are replaced each with fresh threshold signature keys $\ski',\pki'$ sampled from the same key distribution. Then using these fresh keys, we replace  $(\model,\sig)$ received from $\Fpop$ with locally computed $(\model', \sig')$. 
  $G_2$ is indistinguishable from $G_1$ thanks to threshold signature scheme. Threshold signature scheme ensures that only a threshold-many clients can generate a valid signature (as a proof of participation) ensured by the security against EUF-CMA, and the keys are high entropy giving \Adv a negligible advantage to guess \ski.
  \item \textbf{Game $G_3$:} In this final game, we replace PRF key $\prfkey$ with an independent random key $\prfkey'$ and consequentially $\prfvalue'$ is generated based on $\prfkey'$ and $\allpk'$. During the \Prove phase, instead of $\allpk'$, \Sim uses \allpk while running OPRF with \Adv. 
  $G_3$ is indistinguishable from $G_2$ due to OPRF security. \Adv needs to break the OPRF security in order to learn if \Sim used $\allpk'$. Hence, OPRF ensures the pseudorandom key cannot be distinguished by \Adv. 
\end{itemize}

Given the games above, $ G_0 \approx_c G_1 \approx_c G_2 \approx_c G_3$ proves that the real-world execution is computationally indistinguishable from the ideal world. Hence, with this proof, we show that our \FedPoP instance satisfies the soundness requirement.

\begin{theorem}\label{thm:malicious_sp}
Our \FedPoP is secure according to Definition \ref{def:pomsecdef} against any non-uniform PPT adversary \Adv corrupting the service provider \posserver assuming that the threshold signature is secure, the oblivious pseudorandom function is secure, and the hash function is collision-resistant.
\end{theorem}
\begin{proof} 
The simulator \Sim simulates the honest parties which are the \flserver and  \clientset in the real world and malicious parties in the ideal world which is \posserver. Note that \Sim stores all the data it receives, generates, and sends in its database $DB$. \Sim behaves as follows: \\
\noindent \textbf{Setup Phase:}
\begin{itemize}
    \item \Sim receives $(\model,\globalpublicwitness=\{\allpk,\prfvalue\})$ from \Fpop.
    \item \Sim generates data for each honest client it is simulating as $\{\idata'\}_{i\in[n]}$ from the same data distribution of $\idata$, and the public parameter $\mlparam'$ from FL public parameter distribution and generates $\langle \{\ski',\pki'\}_{i\in [n]},\allpk'\rangle$ by choosing the keys from the key distributions of the algorithms of threshold signature scheme, and secure aggregation in the real world. 
\end{itemize}
\textbf{Generate Phase:}
\begin{itemize} 
    \item Upon receiving  $\langle \model, \allpk, \prfvalue \rangle$ from \Fpop, \Sim computes $\Train(\mlparam',\idata')\rightarrow \mlmsg'$ for each clients and aggregates all the $\mlmsg'$ to generate $\model'$. 
    \item \Sim generates a signature $\TSSign(\model',\ski')\rightarrow \sig_i'$ and then $\TSSignAgg(\{\sig_i'\}_{i\in[n]},\allpk')\rightarrow \sig'$. 
    \item \Sim generates a PRF key $\prfkey'$ via OPRF key generation algorithm and  computes $\prfvalue'=F_{\prfkey'}(\allpk')$.
    \item \Sim sends $\langle \model',\allpk',\prfvalue' \rangle$ to \Adv. 
    \item \Sim stores all the values in $DB$.
\end{itemize}
\textbf{Prove Phase:}
\begin{itemize}
     \item \Sim sends $\hash(\model')$ to \Adv. 
     \item If \Adv returns abort, \Sim sends $\langle\model',\allpk\rangle$ to \Fpop. Otherwise, \Adv sends $\tilde{\alpha}$ to \Sim. 
     \item \Sim checks if $\tilde{\alpha} \in \G$. If it is, \Sim computes $\beta = \tilde{\alpha}^{\prfkey'}$. 
     \item \Sim sends $\langle \sig,\beta\rangle$ to \Adv. 
     \item If \Adv returns $0$, \Sim sends $\langle\model',\allpk'\rangle$ to \Fpop. Otherwise, \Sim sends $\langle\model,\allpk\rangle$ to \Fpop.
\end{itemize}
\end{proof}

\begin{claim}
    The view of \Adv, corrupting \posserver, in their interaction with \Sim is indistinguishable from the view in their interaction with the real-world honest parties. 
\end{claim}
\begin{proof}    
Similar to the previous proof, we also prove this claim via a sequence of hybrid games. 
Note that if \Adv distinguishes its interaction with \Sim from the real protocol, then \Adv must violate any of the security requirements: 1) privacy; 2) anonymity; or 3) unlinkability. 
The details of the games with their reductions are presented below:
\begin{itemize}
 \item \textbf{Game $G_0$:} In this game, we use the inputs of honest parties as the inputs of our simulation. This game is identical to real-world \FedPoP protocol. 
\item \textbf{Game $G_1$:} This game is the same as $G_0$ except when running the SA protocol, \Sim used a different \mlparam, \idata, $\sarandb$, and $\sapairwisemask$ which are respectfully derived from their corresponding distribution, and it resulted in $\model'$. 
\Adv can distinguish this behavior from $G_0$ only if it can break the secure aggregation confidentiality or the collusion resistant hash function. However, secure aggregation security hides the true values using random key values and prevents \Adv from learning $\idata'$ which is used instead of true data while hash function ensures that \Adv cannot cheat while sending $\tilde{h}$ since the only way for \Adv to succeed is to break the security of hash function. Hence, $G_1$ is computationally indistinguishable from $G_0$. 
\item \textbf{Game $G_2$:} $G_2$ is the same as the previous game except that \Sim behaves differently while resampling $\{(\sk_i',\pk_i')\}_{[n]}$ and $\allpk'$ from their public distributions and redirect subsequent messages accordingly; training/aggregation remain honest. Since key-generation is random and messages depend only on published keys, $G_2$ is computationally indistinguishable from $G_2$. 
\item \textbf{Game $G_3$}:
$G_3$ is the same as the previous game except that \Sim simulates $\prfkey$ by choosing $\prfkey'$ from the same OPRF key distribution and sending $\prfvalue'$ calculated on $\allpk'$ and $\prfkey'$ to \Adv. 
\Sim replaces $\prfkey$ by $\prfkey'$, set $\prfvalue'=F_{\prfkey'}(\allpk')$, and answers challenges by $\beta=\tilde{\alpha}^{\prfkey'}$. 
If $\Adv$ can distinguish $G_3$ from $G_2$, then we build a PPT distinguisher that breaks OPRF security. 
Therefore, $G_3$ is computationally indistinguishable from $G_2$ under OPRF security. 
\item \textbf{Game $G_4$:}
$G_4$ is the same as $G_3$ except that 
\Sim produces $\model'$, $\hash(\model')$, and $\sig'$ from public distributions and send $\langle \model',\allpk',\prfvalue'\rangle$ and later $\langle \sig',\beta\rangle$. \sig is the unique group signature on $\hash(\model)$ under $\allpk$ (this is the standard subset-hiding property of threshold signature schemes). 
As inputs and algorithms match the distributions and honest evaluation, $G_4$ is computationally indistinguishable from $G_3$. 
\end{itemize}

\emph{Anonymity} is single-round and asks the adversary to identify the prover among $C_0,C_1$; a constructed reduction $\B$ embeds the signer-indistinguishability challenge in that round by choosing $\clientset_0,\clientset_1$ identical except for including $C_0$ vs.\ $C_1$. 
\emph{Unlinkability} is two-round and asks whether two transcripts come from the same honest client or two distinct honest clients; $\B$ uses a random-round embedding (one of the two rounds carries the challenge signature, the other is simulated identically in both worlds). 
\emph{Privacy} fixes one honest prover $\clienti$ and asks that the server’s view not reveal which other clients participated; the game compares two authorized subsets $\clientset_0,\clientset_1$ of equal size $\ge t$ that both contain $\clienti$, and $\B$ challenges on $(m,\clientset_0,\clientset_1)$, allowing only the leakage that ``at least $t$ participated.'' 
Hence, in our model \emph{unlinkability implies anonymity} (an identifier would yield a linker), while \emph{privacy is orthogonal} to both: anonymity need not hide co-signer membership and privacy need not hide which single client proved. All three ultimately reduce to the $G_3 \to G_4$ via signer-indistinguishability.

\begin{lemma}[Anonymity]
For any two honest clients $C_0$ and $C_1$ and any round $l$, the view of \Adv when the prover is $C_0$ is computationally indistinguishable from the view when the proving client is $C_1$. 
\end{lemma}
\begin{proof} 
For the proof of anonymity, we consider the standard indistinguishability experiment that selects one of clients $C_0$ and $C_1$ uniformly at random as the prover. The hybrids $G_0 \to G_2$ are independent of this choice. 
In $G_3$, the only identity-bearing artifacts are the OPRF tuples $(\prfvalue,\beta)$, but these are either real evaluations under a fresh/hidden key or pseudorandom and, in either case, independent of the client's identity; otherwise, \Adv (as a distinguisher) would break OPRF security. 
In $G_4$, the exposed proof of participation is $\sig'$, which (by signer-indistinguishability) depends only on $m=\hash(\model)$. The signature distribution does not depend on which authorized subset produced partial signatures. 
Now, let us suppose, toward contradiction, that an adversary $\Adv$ identifies the prover with non-negligible advantage $\epsilon(\lambda)$ after the $G_0\to G_3$ hops. We build a reduction $\B$ that breaks the signer-indistinguishability of the threshold signature scheme. The challenger provides $\allpk$ and a signing oracle that, on input $(m,\clientset)$ for any authorized subset $\clientset$ of size $t$, returns a full signature on $m$. $\B$ simulates the entire round view for $\Adv$ using uniform OPRF outputs and its signing oracle for any auxiliary signatures. For the challenge embedding, $\B$ sets $m^\star = \hash(\model_l),\allpk$ and chooses authorized subsets $\clientset_0,\clientset_1$ that are identical except that $\clientset_0$ includes $C_0$ while $\clientset_1$ includes $C_1$ (both of size $t$). $\B$ submits $(m^\star,\clientset_0,\clientset_1)$ and receives $\sig^\star$ for a hidden bit $b$, then uses $\sig^\star$ as the round-$l$ signature in the simulated view. If $\Adv$ succeeds, it distinguishes whether $\sig^\star$ was generated by $\clientset_0$ or by $\clientset_1$, which gives $\B$ advantage $\epsilon(\lambda)-\negl$ for the threshold signature scheme's security game. This contradicts the assumption that the threshold signature scheme is signer-indistinguishable against polynomial-time adversaries. Therefore, \FedPoP satisfies the anonymity requirement. \footnote{Recall that a client can communicate with \posserver via an anonymous communication in order to hide its identifier (IP address) and this also applies to unlinkability.}
\end{proof}

\begin{lemma}[Unlinkability] Across two rounds, the joint view of $\Adv$ is indistinguishable between the case where the same honest client proves twice and the case where two distinct honest clients prove. \end{lemma}
\begin{proof}
For the proof of unlinkability, we run the hybrids $G_0 \to G_4$ independently for each round $l$. 
$G_0 \to G_3$. Each round has fresh OPRF keys. 
By multi-instance OPRF security, $(\prfvalue_l,\beta_l)$ are indistinguishable from values chosen from the uniform distribution over their ranges, independently across rounds and independently of the proving client’s identity. 
Therefore, only the threshold signature could carry linkage.
$G_3 \to G_4$. In round $l$, the full signature $\sig_l$ is generated on $m_l=\hash(\model_l)$. 
By signer-indistinguishability, the distribution of $\sig_l$ depends only on $m_l$, not on which authorized subset of clients produced the partial signatures (which reveals nothing about which subset of clients $\{\clienti\}$ participated in that round). 

Now, let us assume, toward contradiction, that an adversary $\Adv$ wins the unlinkability game with non-negligible advantage $\epsilon(\lambda)$ after $G_0\to G_3$; i.e., given two challenge transcripts for rounds $0$ and $1$, it decides whether they came from the same honest client or from two distinct honest clients better than random. We build a reduction $\B$ that breaks the signer-indistinguishability of the threshold signature scheme. 
The challenger provides $\allpk$ and a signing oracle for any authorized subset $\clientset$ of size $t$. $\B$ simulates both rounds for $\Adv$ using uniform OPRF outputs and its signing oracle for all non-challenge signatures. To embed the challenge, $\B$ chooses a random round $l^\star\in\{0,1\}$, sets $m^\star \gets \langle l^\star,\, \hash(\model_{l^\star}), \allpk \rangle$, and chooses authorized subsets $\clientset_0,\clientset_1$ that are identical except for swapping the honest client whose linkage $\Adv$ attempts to detect. 
$\B$ submits $(m^\star,\clientset_0,\clientset_1)$ and receives $\sig^\star$ for a hidden bit $b$, then uses $\sig^\star$ as the signature for round $l^\star$; the other round is signed via the oracle in a way that is identical under both cases (``same client twice'' vs. ``two distinct clients''). 
If $\Adv$ succeeds, it distinguishes whether $\sig^\star$ was generated by $\clientset_0$ or by $\clientset_1$, which gives $\B$ advantage at least $\epsilon(\lambda)-\negl$ in the signer-indistinguishability game of the threshold signature scheme. This contradicts the assumption that the threshold signature scheme is signer-indistinguishable. Therefore, \FedPoP satisfies unlinkability. 

\end{proof}

\begin{lemma}[Privacy]
Let $\clientset_0,\clientset_1\subseteq \clientset$ be two subsets with $|\clientset_0|=|\clientset_1|\ge t$, both containing the (honest) prover. 
Then the views of $\posserver$ in runs where the participating set is $\clientset_0$ vs.\ $\clientset_1$ are computationally indistinguishable, beyond the permitted leakage of the fact ``at least $t$ clients participated.'' \end{lemma}
\begin{proof}
For the proof of privacy, we first define hybrids $G_0\to G_4$ for a single run. 
Up to $G_2$, the view depends only on the fact that the threshold condition is met (i.e., $|\clientset_b|\ge t$), not on which particular clients are in the set, because the messages observed by $\posserver$ are either public or masked/aggregated. 
In $G_3$, the OPRF tuples $(\prfvalue,\beta)$ depend on the public input and hidden key, not on the actual subset; otherwise we could break OPRF security. 
In $G_4$, $\sig'$ leaks no membership information. 
Finally, note that $\allpk$ (the aggregation of clients' public keys) carries no client-identifying linkage information (no mapping to $\pki$); it is invariant under permutations of client identities and thus reveals no membership. 
Therefore $G_0$ and $G_4$ are computationally indistinguishable, and so the two distributions of $\Adv$'s view for $\clientset_0$ and $\clientset_1$ are indistinguishable. 
For a concise reduction, assume an adversary $\Adv$ distinguishes the $\posserver$-views for $\clientset_0$ versus $\clientset_1$ with non-negligible advantage $\epsilon(\lambda)$ after $G_0 \to G_3$. We build $\B$ against signer-indistinguishability: $\B$ simulates the run for $\Adv$ using uniform OPRF outputs and its signing oracle for any auxiliary signatures; it sets $m^\star \gets \hash(\model)$ and submits $(m^\star,\clientset_0,\clientset_1)$ to the challenger, receiving $\sig^\star$. It embeds $\sig^\star$ as the run’s threshold signature in the simulated view. If $\Adv$ succeeds, it distinguishes whether $\sig^\star$ was generated by $\clientset_0$ or by $\clientset_1$, which gives $\B$ advantage $\epsilon(\lambda)-\negl$ in the signer-indistinguishability game of the threshold signature scheme. This contradicts the assumption that the threshold signature scheme is signer-indistinguishable against polynomial-time adversaries. 
Therefore, up to the permitted leakage that at least $t$ clients participated, the views for $\clientset_0$ and $\clientset_1$ are computationally indistinguishable, as claimed. Hence, \FedPoP satisfies privacy. 
Hence, \FedPoP satisfies privacy.
\end{proof}

\end{proof}

\subsection{Performance Analysis}\label{performance}
\subsubsection{Setup}\label{setup}
We implement a proof of concept of \FedPoP in \texttt{Python} on a machine with Intel Core i5 at 2.3 GHz and 16GB RAM. 
We use the FL framework $\mathsf{Flower}$ \cite{beutel2020flower} to implement our FL design.  
For the secure aggregation, we apply the $\mathsf{SecAgg+}$ secure aggregation protocol \cite{bell2020secure} which is based on masking. 
We adopt $\mathsf{liboprf}$ \cite{liboprf} for the OPRF implementation and $\mathsf{ROAST}$ \cite{roastPaper} for the threshold signature. 
The numbers reported for the experimental results are an average of 10 measurements.

\noindent\textbf{DNN Structure and Data Split}. 
We evaluated our approach using the \texttt{MNIST} dataset, consisting of 70K grayscale images with 10 class labels (corresponding to 10 digits), in which 60K images are for training and 10K images are for testing.  
The training set was partitioned among clients (with client counts of 10, 25, 50, and 100) using a non-independent and identically distributed (non-i.i.d.) label skew strategy. 
The testing set was used uniformly by all clients for model evaluation. 
The model architecture is a lightweight CNN ($\mathsf{SimpleCNN}$) comprising two convolutional layers (16 and 32 filters respectively), each followed by ReLU activation and 2×2 max pooling, and two fully connected layers leading to a 10-class output. 
Each client trained locally for 10 epochs per communication round using the Adam optimizer (learning rate 0.0005) and Cross-Entropy Loss. 
We simulated client dropouts with rates of $10\%, 30\%, 50\%$ and $70\%$. 
A randomly selected subset of active clients based on dropout rates trained locally, masked their updates, and sent them to the server for secure aggregation followed by signing the received \model and \prfkey generation. 
For masking, the masking values are randomly selected from a normal distribution.
For threshold signature, the threshold $t$ represents the dropout tolerance. For instance, for $n=100$ and $10\%$ dropout rate, $t$ is $90$ ($t=100\times0.9$), which means that at least $90\%$ of the clients shall remain online. 
Performance metrics included client computation time, communication overhead, and server aggregation time.

\subsubsection{Evaluation}
We evaluate \FedPoP in terms of the communication and computation overhead introduced during the \Setup, \Generate and \Prove phases. 

\noindent \textbf{Setup.} For \Setup, the only overhead is from the threshold signature initialization to the securely aggregated FL, which can take up to $\sim 33.6$ seconds for 100 clients when $t=50$. 
Table \ref{tab:ts_setup} presents the setup time required for threshold signature across varying number of clients and threshold percentages. 
The setup time increases significantly with the number of signers due to the computational and communication overhead introduced by larger group sizes. For instance, at a 50\% threshold, setup time grows from approximately 0.0597 seconds with 10 signers (where at least 5 clients must be online to sign) to over 33 seconds with 100 signers (where at least 50 clients must be online to sign). For $n\le25$, the cost remains sub-second to under a second. In practice, this one-time cost can be amortized over many rounds. 

\begin{table}[htp]
\centering
\caption{Threshold signature setup time in seconds across different number of clients ($n$) and threshold percentages ($t$).}
\label{tab:ts_setup}
\begin{tabular}{r|cccc}
\toprule
$n$  & $t=30\%$ & $t=50\%$ & $t=70\%$ & $t=90\%$ \\
\midrule
10   & 0.0596 & 0.0597 & 0.0730 & 0.0940 \\
25   & 0.4302 & 0.6062 & 0.8008 & 0.9377 \\
50   & 2.9608 & 4.1044 & 5.5554 & 6.8690 \\
100  & 21.6394 & 33.5832 & 38.6274 & 49.2601 \\
\bottomrule
\end{tabular}
\end{table}

\begin{figure}
\hspace*{-1cm}
\begin{tikzpicture}
\begin{axis}[
    xlabel={Threshold $t$ (\%)},
    ylabel={Time (s)},
    legend style={at={(0.02,0.98)}, anchor=north west, draw=none, fill=white, font=\footnotesize},
    grid=both,
    width=9.5cm,
    height=5.5cm,
    mark options={solid},
    nodes near coords,
    point meta=explicit symbolic,
    every node near coord/.append style={font=\scriptsize, anchor=south},
    line width=1.2pt
]

\addplot+[mark=*, color=plotblue] coordinates {
    (30,0.0596) 
    (50,0.0597) 
    (70,0.0730) 
    (90,0.0940) 
};
\addlegendentry{$n=10$}

\addplot+[mark=square*, color=plotred] coordinates {
    (30,0.4302) 
    (50,0.6062) 
    (70,0.8008) 
    (90,0.9377) 
};
\addlegendentry{$n=25$}

\addplot+[mark=triangle*, color=plotgreen] coordinates {
    (30,2.9608) 
    (50,4.1044) 
    (70,5.5554) 
    (90,6.8690) 
};
\addlegendentry{$n=50$}

\addplot+[mark=diamond*, color=plotpurple] coordinates {
    (30,21.6394) 
    (50,33.5832) 
    (70,38.6274) 
    (90,49.2601) 
};
\addlegendentry{$n=100$}
\end{axis}
\end{tikzpicture}
\caption{Client threshold signature setup time in seconds.}
\end{figure}
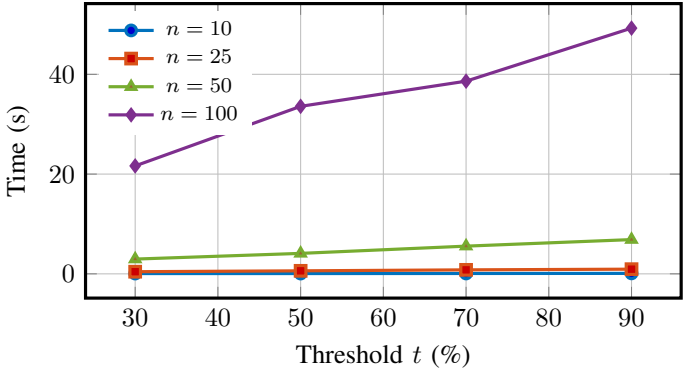

\begin{table}[htb]
\centering
\caption{Average execution time of the components of \Generate in seconds. } 
\begin{tabular}{c c c|c|c}
 \textbf{$n$}  & \textbf{$n_{drop}$}& \textbf{$t$}&\textbf{Client}
& \textbf{FL Server} \\
\hline
\multirow{4}{*}{\rotatebox{45}{$n=$\textbf{10}}}  & $0.1n$   & $ 0.9n$     & \textbf{65.1968 }                        & 0.6623\\
                         & $0.3n$    & $ 0.7n$   &       49.2949                         &  0.2233\\
                         & $0.5n$    & $ 0.5n$    &  32.6239                         & 0.0257\\ 
                         & $0.7n$    & $ 0.3n$   &       \textbf{18.7537    }                      & 0.0217\\ \hline
\multirow{4}{*}{\rotatebox{45}{$n=$\textbf{25}}}   & $0.1n$   & $ 0.9n$     &    \textbf{63.3622   }                      & 0.1475\\
                         & $0.3n$   & $ 0.7n$    &         47.9265                       &  0.0840\\
                         & $0.5n$    & $ 0.5n$    &       35.6831                          &  0.0556\\ 
                         & $0.7n$    & $ 0.3n$   &       \textbf{23.0483}                   & 0.0408\\ \hline
                         
\multirow{4}{*}{\rotatebox{45}{$n=$\textbf{50} }}   & $ 0.1n$  & $ 0.9n$      &   \textbf{67.3044  }                       & 0.3275\\
                         & $0.3n$    & $ 0.7n$   &                52.7199                          & 0.2414\\
                         & $0.5n$   & $ 0.5n$     &      37.9622                          & 0.1253\\ 
                         & $0.7n$  & $ 0.3n$     &                    \textbf{21.0227 }                        & 0.0970\\ \hline
\multirow{4}{*}{\rotatebox{45}{$n=$\textbf{100}}}  & $0.1n$   & $ 0.9n$     &     \textbf{96.9902   }                       &  1.1668\\
                         & $0.3n$   & $ 0.7n$    &      73.3057                          & 0.7930\\
                         & $0.5n$  & $ 0.5n$      &        50.1172                         &  0.5186\\ 
                         & $0.7n$   & $ 0.3n$    &       \textbf{29.4788}                       &   0.2668 \\ \hline
\end{tabular} 
\label{tab:avg-exectime}
\end{table}

\begin{figure} 
\hspace*{-0.3cm}
\begin{tikzpicture}
\begin{axis}[
    xlabel={Dropout $n_{drop}$ (\%)},
    ylabel={Time (s)},
    legend style={at={(1,0.98)}, anchor=north east, draw=none, fill=white, font=\footnotesize},
    grid=both,
    width=9cm, 
    height=5.5cm,
    mark options={solid},
    nodes near coords,
    point meta=explicit symbolic,
    every node near coord/.append style={font=\scriptsize, anchor=south},
    line width=1.2pt
]

\addplot+[mark=*, color=plotblue] coordinates {
    (10,0.6560)
    (30,0.2190) 
    (50,0.0228) 
    (70,0.0192) 
};
\addlegendentry{$n=10$}

\addplot+[mark=square*, color=plotred] coordinates {
    (10,0.1260) 
    (30,0.0695) 
    (50,0.0460) 
    (70,0.0357) 
};
\addlegendentry{$n=25$}

\addplot+[mark=triangle*, color=plotgreen] coordinates {
    (10,0.2646) 
    (30,0.1957) 
    (50,0.0986) 
    (70,0.0798) 
};
\addlegendentry{$n=50$}

\addplot+[mark=diamond*, color=plotpurple] coordinates {
    (10,0.9700) 
    (30,0.6471) 
    (50,0.4321) 
    (70,0.2100) 
};
\addlegendentry{$n=100$}

\end{axis}
\end{tikzpicture}
 \caption{FL Server-side computation time for secure aggregation in seconds.}
    \label{fig:flserverSA}
    \end{figure}
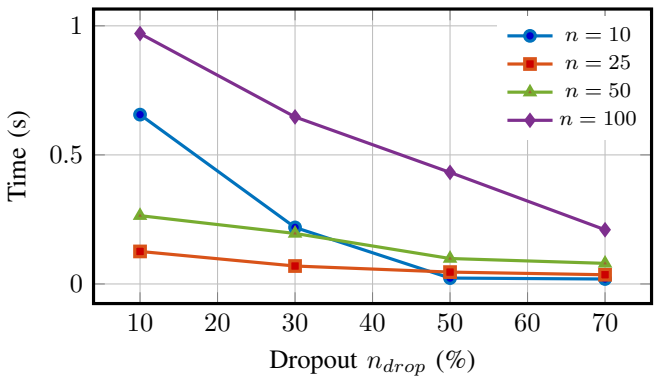

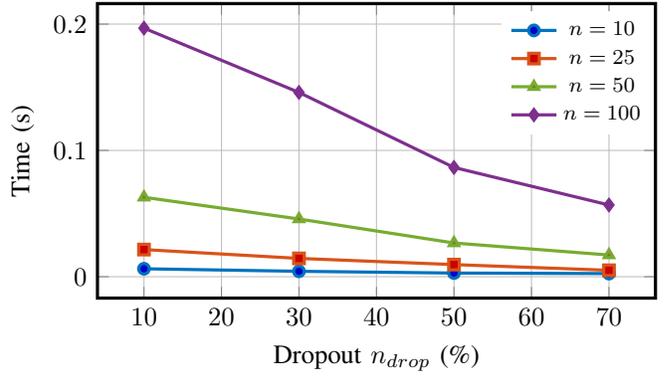
\begin{figure} 
\hspace*{-0.3cm}
\begin{tikzpicture}
\begin{axis}[
    xlabel={Dropout $n_{drop}$ (\%)},
    ylabel={Time (s)},
    legend style={
        at={(1,0.98)},
        anchor=north east,
        draw=none,
        fill=white,
        font=\footnotesize,
        cells={align=left}
    },
    grid=both,
    width=9cm, 
    height=5.5cm,
    mark options={solid},
    nodes near coords,
    point meta=explicit symbolic,
    every node near coord/.append style={font=\scriptsize, anchor=south},
    line width=1.2pt
]

\addplot+[mark=*, color=plotblue] coordinates {
    (10,0.0063) 
    (30,0.0043)
    (50,0.0029) 
    (70,0.0025) 
};
\addlegendentry{$n=10$}

\addplot+[mark=square*, color=plotred] coordinates {
    (10,0.0215) 
    (30,0.0145) 
    (50,0.0096) 
    (70,0.0051) 
};
\addlegendentry{$n=25$}

\addplot+[mark=triangle*, color=plotgreen] coordinates {
    (10,0.0629) 
    (30,0.0457) 
    (50,0.0267) 
    (70,0.0172) 
};
\addlegendentry{$n=50$}

\addplot+[mark=diamond*, color=plotpurple] coordinates {
    (10,0.1968) 
    (30,0.1459)
    (50,0.0865) 
    (70,0.0568) 
};
\addlegendentry{$n=100$}
\end{axis}
\end{tikzpicture}
    \caption{FL Server-side computation time for threshold signature aggregation in seconds.}
    \label{fig:flserverTSAgg}
\end{figure}

\noindent\textbf{Generate.}\label{computational_overhead}
We measure the execution time of each party (clients and server) in the \Generate phase. 
The total execution time is reported in Table~\ref{tab:avg-exectime} for various values of $n$ and $t$. Complete per-component numerical results are provided in Appendix~\ref{app:full_num_generate} by Table~\ref{tab:exetime}. 
The execution time increases with the ratio $t/n$, as higher thresholds require more clients to participate. For example, with $n=100$, $t=0.9n$, and $n_{\mathit{drop}}=0.1n$, \FedPoP adds only $0.97$ seconds of overhead to the underlying secure aggregation protocol. 
In terms of FL server-side computation, we evaluated the secure aggregation and threshold signature aggregation performed by the FL server. Fig. \ref{fig:flserverTSAgg} demonstrates that when the number of clients increases, so the number of partial signatures to aggregate, the signature  aggregation times increases. Fig. \ref{fig:flserverSA} shows the effect of drop-outs on secure aggregation for various number of clients. When dropout percentage increases, meaning that the number of online clients decreases, the computation time of secure aggregation decreases.

In terms of client-side computation, we evaluated the securely aggregated training and threshold signature generation performed by clients. 
Fig. \ref{fig:clientSAtraining} and Fig. \ref{fig:clientTSGen} plot wall-clock time as the dropout rate varies over 10, 30, 50, and 70 percent for client group sizes 10, 25, 50, and 100. 
Fig. \ref{fig:clientSAtraining} shows that when the number of dropped clients increases, the time for securely aggregated training decreases for every group size. The decrease is large: for 10 clients the time goes from 65.19 seconds at 10 percent dropout to 18.75 seconds at 70 percent; for 25 clients it goes from 63.34 to 23.04 seconds; for 50 clients from 67.24 to 21.00 seconds; and for 100 clients from 96.79 to 29.40 seconds.
Fig. \ref{fig:clientTSGen} shows that signature generation also becomes faster as dropout increases, since fewer signing shares are required. The absolute cost is in milliseconds and remains small compared with secure aggregation. 
For small sets of clients such as 10 and 25 clients the variation is minor and the total time is negligible compared with secure aggregation. For larger size of clients the scaling is close to linear in the number of online clients.

\begin{figure}[htp] 
\hspace*{-0.6cm}
\begin{tikzpicture}
\begin{axis}[
    xlabel={Dropout $n_{drop}$ (\%)},
    ylabel={Time (s)},
    legend style={at={(1,0.98)}, anchor=north east, draw=none, fill=white, font=\footnotesize},
    grid=both,
    width=9.5cm, 
    height=5.5cm,
    mark options={solid},
    nodes near coords,
    point meta=explicit symbolic,
    every node near coord/.append style={font=\scriptsize, anchor=south},
    line width=1.2pt
]

\addplot+[mark=*, color=plotblue] coordinates {
    (10,65.19) 
    (30,49.29)
    (50,32.62)
    (70,18.75) 
};
\addlegendentry{$n=10$}

\addplot+[mark=square*, color=plotred] coordinates {
    (10,63.34) 
    (30,47.91) 
    (50,35.67) 
    (70,23.04) 
};
\addlegendentry{$n=25$}

\addplot+[mark=triangle*, color=plotgreen] coordinates {
    (10,67.24)
    (30,52.67) 
    (50,37.93) 
    (70,21.00) 
};
\addlegendentry{$n=50$}

\addplot+[mark=diamond*, color=plotpurple] coordinates {
    (10,96.79)
    (30,73.16) 
    (50,50.01) 
    (70,29.40) 
};
\addlegendentry{$n=100$}

\end{axis}
\end{tikzpicture}
 \caption{Client-side computation time for securely aggregated training in seconds.}
    \label{fig:clientSAtraining}
    \end{figure}
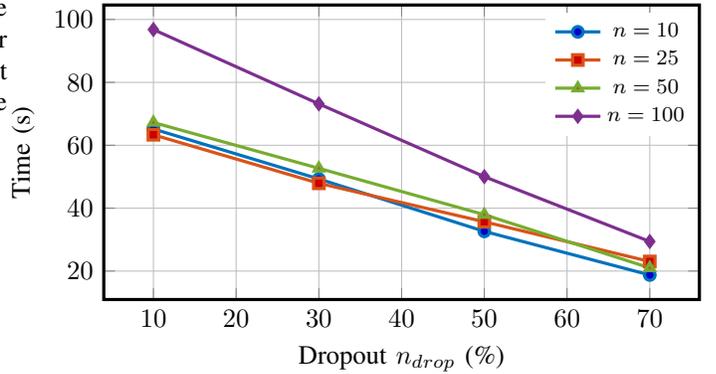

\begin{figure}
\hspace*{-0.5cm}
\begin{tikzpicture}
\begin{axis}[
    xlabel={Dropout $n_{drop}$ (\%)},
    ylabel={Time (s)},
    legend style={
        at={(1,0.98)},
        anchor=north east,
        draw=none,
        fill=white,
        font=\footnotesize,
        cells={align=left}
    },
    grid=both,
    width=9.5cm, 
    height=5.5cm,
    mark options={solid},
    nodes near coords,
    point meta=explicit symbolic,
    every node near coord/.append style={font=\scriptsize, anchor=south},
    line width=1.2pt
]

\addplot+[mark=*, color=plotblue] coordinates {
    (10,0.0068) 
    (30,0.0049) 
    (50,0.0039) 
    (70,0.0037) 
};
\addlegendentry{$n=10$}

\addplot+[mark=square*, color=plotred] coordinates {
    (10,0.0222) 
    (30,0.0165) 
    (50,0.0131) 
    (70,0.0083) 
};
\addlegendentry{$n=25$}

\addplot+[mark=triangle*, color=plotgreen] coordinates {
    (10,0.0644) 
    (30,0.0499) 
    (50,0.0322) 
    (70,0.0227) 
};
\addlegendentry{$n=50$}

\addplot+[mark=diamond*, color=plotpurple] coordinates {
    (10,0.2002) 
    (30,0.1457) 
    (50,0.0972) 
    (70,0.0788) 
};
\addlegendentry{$n=100$}
\end{axis}
\end{tikzpicture}
    \caption{Client-side computation time for threshold signature generation in seconds.}
    \label{fig:clientTSGen}
\end{figure}
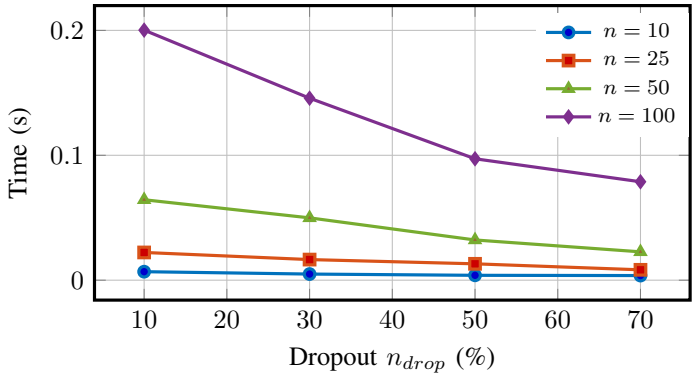

\begin{figure}
\centering
\begin{tikzpicture}
\begin{axis}[
    ybar,
    bar width=20pt,
    enlarge x limits=0.5,
    ylabel={Time (s)},
    symbolic x coords={Client, Service Provider},
    xtick=data,
    nodes near coords,
    width=8cm,
    height=6cm,
    ymin=0,
    ymax=0.07,
    scaled y ticks=false, 
    y tick label style={/pgf/number format/fixed}, 
    nodes near coords,
    point meta=explicit symbolic
]
\addplot coordinates {(Client, 0.0052) [0.0052] (Service Provider, 0.0560) [0.0560]};
\end{axis}
\end{tikzpicture}
\caption{Computation time during the Prove phase of \FedPoP for the client and the service provider, measured in seconds.}
\label{fig:fedpop_prove_execution}
\end{figure}
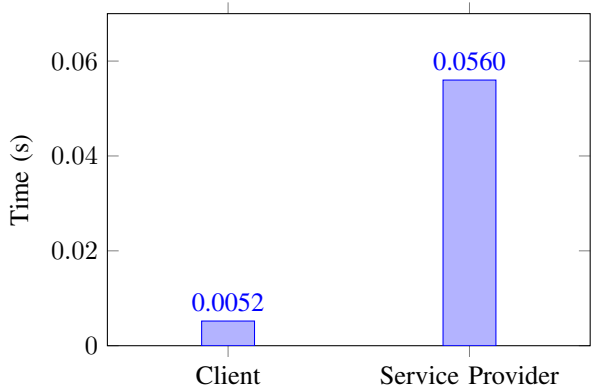

\noindent \textbf{Prove.} The \Prove phase involves only signature verification and OPRF computation while we exclude hash function evaluation and comparison due their negligible computation cost. As shown by Fig. \ref{fig:fedpop_prove_execution}, computing $\beta$ takes approximately $0.0052$ seconds, while the OPRF evaluation and signature verification take $0.0558$ and $0.0002$ seconds, respectively. In total, generating a proof of participation takes on average $0.0612$ seconds, introducing minimal client-side overhead.

\noindent \textbf{Communication Overhead.}\label{communication_overhead}
We evaluate communication overhead by measuring the data exchanged in both the \Generate and \Prove phases. \FedPoP introduces one additional communication round over standard secure aggregation because clients share partial signatures with \flserver and receive the aggregated signature together with \prfkey. This overhead is minimal due to the small size of the signature and \prfkey relative to the model size (approximately 810KB in our experiments). During the \Prove phase, \posserver sends about 95 bytes and receives 315 bytes from the client. Although the results indicate low overhead, we plan to compare \FedPoP's efficiency more thoroughly against related work~\cite{pflm21}.

\section{Further Discussion}\label{discussion}
In this section, we propose possible adjustments to \FedPoP for more sophisticated properties and discuss some open challenges.

\noindent \underline{\textbf{Key Management}}. One further challenge is the management of clients' keys. 
In practice, clients may contribute to many different models with different FL servers. 
Therefore, running the threshold signature setup could be cumbersome for the clients for two reasons: 1) storage of these keys; and 2) security of the keys. 
To avoid this the concept of derivable asymmetric keys can be adopted. 
However, note that PRF key has to be generated freshly for each model generation. 
This approach is proposed as the key derivation scheme used in Bitcoin Improvement Proposal (BIP) 32~\cite{bip32}. 
The key derivation scheme allows a chain of child public keys (i.e., clients' verification keys used during \Generate) to be derived from a single parent public key. 
\prfkey and \ski shall be secured.  
To securely store these keys, a secure storage on a single server or distributed multi-server can be deployed.  

\noindent \underline{\textbf{Malicious \posserver}}. Our approach considers an honest-but-curious \posserver. 
However, \FedPoP can be secured against a malicious \posserver. 
To do so, an honest \clienti has to verify if every computation \posserver performs is computed correctly. 
This can be achieved by more advanced verifiable computations such as zero-knowledge proof or by integrating verifiable federated learning approaches \cite{verifynet}. 
This requires, however, further investigation, and we leave it as a future work.

\noindent \underline{\textbf{Applicability.}} \FedPoP generates proof of participation after the model is generated using secure aggregation. Therefore, our \FedPoP instance can be integrated to any securely-aggregated FL designs. 
However, the overall efficiency is determined by the underlying aggregation protocol. A comprehensive empirical comparison across protocols is deferred to future work. 

\noindent \underline{\textbf{Multiple Rounds and Proof Generation.}} In principle, FL composes of multiple rounds in which many consecutive summations of model weights are performed to derive a (converged) model. A client may or may not participate to every round. Therefore, it is important that in each round a proof is generated along with \model. However, this increases the number of proofs and witnesses
on the client side. Thus, it is interesting to investigate in the future how to accumulate the proofs and witnesses. 
\section{Conclusion and Future Work} 
We proposed \FedPoP, a privacy preserving FL framework with proof of participation. 
\FedPoP ensures that the client's identity is kept secret and that the proof of participation does not reveal the identities of other clients. The service provider cannot link multiple participations of the same client to the model. 
\FedPoP utilizes threshold signatures to enable clients to generate a signature on a global model even in the presence of dropped out clients and employs OPRF to prove the knowledge of the secret witness without expensive computations. 
Addressing key management of secrets, protecting \FedPoP against malicious servers/service providers, and analyzing practical attacks on \FedPoP using an FL benchmark for attacks and defenses \cite{fedsecurity24} are intriguing future directions. 

\section*{Acknowledgements}
The authors thank Prof. Gene Tsudik for his invaluable comments and discussions. This work was started and partially conducted while Devriş İşler was visiting Prof. Gene Tsudik at the University of California, Irvine and contributing to the ProperData project, also supported by NSF Award 1956393, and Seoyeon Hwang was affiliated with the University of California, Irvine. 
Devriş İşler was supported by the European Union’s HORIZON project DataBri-X (101070069). Nikolaos Laoutaris was supported by the MLEDGE project (TSI-063100-2022-0004), funded by the Ministry of Economic Affairs and Digital Transformation and the European Union NextGenerationEU/PRTR. 


\textbf{Author contribution}.
\textbf{Devriş İşler}: Conceptualization, Formal analysis, Methodology, Project administration, Validation, Visualization, Implementation, Writing - original draft/review \& editing. \textbf{Elina van Kempen}: Conceptualization, Methodology, Writing - review. \textbf{Seoyeon Hwang}: Conceptualization, Methodology, Formal Analysis, Validation, Writing - review \& editing. \textbf{Nikolaos Laoutaris}: Funding acquisition, Review of the final manuscript.

\bibliographystyle{IEEEtran}
\bibliography{sn-bibliography}
\appendix
\subsection{Comparison of Other Approaches}\label{otherinstances} 
As discussed in Section~\ref{motivation}, alternative signature schemes may replace threshold signatures in realizing \Fpop. We consider two relevant group signature types as candidates for \FedPoP: (1) ring signatures and (2) multi-signatures. 

\textit{Ring signatures}~\cite{ringsign1DHKS19} allow a user to sign on behalf of a group without revealing the actual signer. They are ad-hoc and require no central authority. The global verification key is composed of all participant public keys ($\pki$), enabling any member to sign for the group. However, this breaks threshold security, as a malicious client could forge a valid group signature. Moreover, since $\allpk$ discloses participant keys, an honest-but-curious \flserver may identify clients and link their participations. Some schemes, such as~\cite{roastPaper}, avoid this leakage by omitting $\pki$ from verification. 

\textit{Multi-signatures}~\cite{handan2021} involve each signer generating a key pair and jointly producing a compact signature verifiable against all public keys. Like ring signatures, this reveals the full set of participants, violating privacy and unlinkability. 
As a baseline, we also consider a naïve scheme where the \flserver signs the model along with client keys and shares the result. This approach not only reveals identities but also lacks resilience to dropouts and non-repudiation.
\begin{table*}[htb]
\caption{Comparison of other approaches.}
\centering
\begin{tabular}{ |c|c|c|c|c| } 
\hline
\backslashbox{\textbf{\textit{Feature}}}{\textbf{\textit{Approach}}} 
& \textit{Ring Signature-based} & \textit{Multi Signature-based} & \textit{Na\"ive} & \textit{Our \FedPoP}\\ 
& \textit{\FedPoP Instance} & \textit{\FedPoP Instance} &  & \textit{Instance}\\
\hline
Anonymity & {\color{blue}\ding{52}} & {\color{blue}\ding{52}} & {\color{red}\ding{55}} & {\color{blue}\ding{52}}\\ \hline
Privacy & {\color{red}\ding{55}} & {\color{blue}\ding{52}} & {\color{blue}\ding{52}} & {\color{blue}\ding{52}}\\ \hline
Unlinkability & {\color{red}\ding{55}} & {\color{red}\ding{55}}  & {\color{red}\ding{55}} & {\color{blue}\ding{52}}\\ \hline
Drop-out Tolerance & {\color{blue}\ding{52}} & {\color{red}\ding{55}} & {\color{red}\ding{55}} & {\color{blue}\ding{52}}\\ \hline
 \end{tabular}
    \label{tab:comparison}
\end{table*}

\subsection{\FedPoP Instance Group Witness Generation}\label{sec:prfgeneration_alternative}
\begin{figure}[htb]
\hspace*{-2cm}
\includegraphics[width=0.85\textwidth]{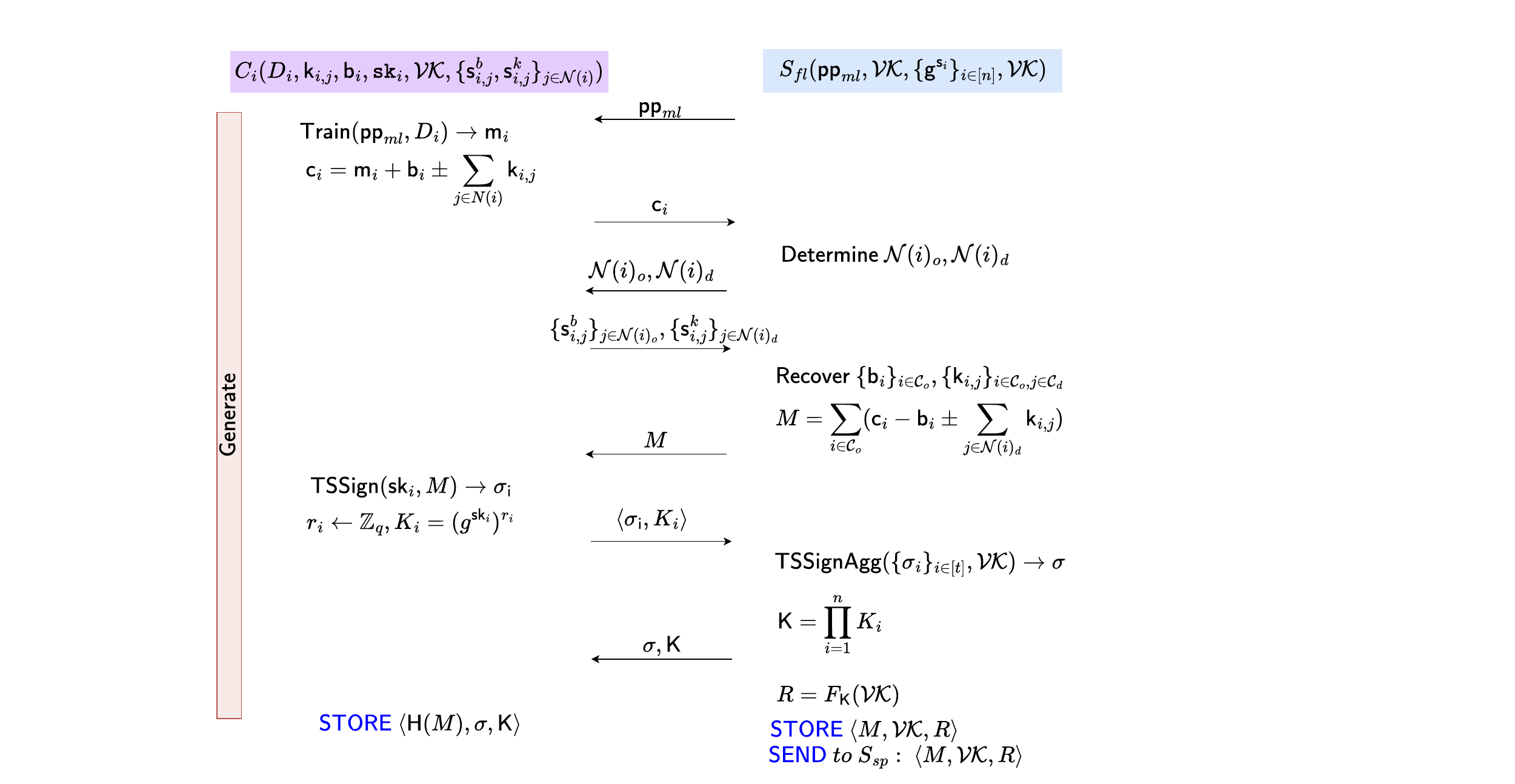}
    \caption{The \Generate phase via alternative \globalwitness creation.}
    \label{fig:alternative_prf_solution}
\end{figure}
In our solution, \flserver generates a PRF key \prfkey to serve as the group witness \globalwitness and distributes it to \onlineclientset. The resulting PRF value is then used as a proof of knowledge of \globalwitness. As noted earlier, \prfkey can be generated in several ways. One alternative is to derive it from \ski, which requires the following change to the \Generate phase: 
Before \clienti sends its partial signature on \model, \clienti samples a random element $ r_i \leftarrow \Zq $, computes $\prfkey_i = (\g^{\ski})^{r_i} $, and sends $ \prfkey_i $ to \flserver. The server then computes the global PRF key from the received shares as $ \prfkey = \prod_{i \in \onlineclientset} \prfkey_i $.\footnote{All operations used to generate \prfkey are performed modulo $q$.}
The key \prfkey serves as the group witness \globalwitness, enabling only participating clients to prove their participation in the \Prove phase as before. Because \prfkey is derived from the clients’ values $ (\g^{\ski})^{r_i} $, only participating clients can obtain it. 
We present the modified protocol in Fig.~\ref{fig:alternative_prf_solution}.

\begin{table*}[htb]
\centering
\caption{Average execution time (in seconds) of the components of \Generate. } 
{\small
\begin{tabular}{c c c|c|c|c|c|c|c}
&&&\multicolumn{3}{c|}{\textbf{Client}} 
& \multicolumn{2}{c}{\textbf{FL Server}} \\
\hline
                       \textbf{$n$}  & \textbf{$n_{drop}$}& \textbf{$t$} & \textbf{SA-Train} & \textbf{$\mathsf{TS}$-Sign} & \textbf{Total} & \textbf{SA-Agg} & \textbf{$\mathsf{TS}$-Agg} & \textit{\textbf{Total}} \\ \hline
\multirow{4}{*}{\rotatebox{45}{$n=$\textbf{10}}}  & $0.1n$   & $ 0.9n$     & 65.19        &           0.0068                  & \textbf{65.1968 }                        & 0.6560 & 0.0063& 0.6623\\
                         & $0.3n$    & $ 0.7n$   &      49.29      &            0.0049           & 49.2949                         & 0.2190& 0.0043& 0.2233\\
                         & $0.5n$    & $ 0.5n$    &     32.62       &           0.0039          & 32.6239                         & 0.0228 & 0.0029& 0.0257\\ 
                         & $0.7n$    & $ 0.3n$   &       18.75     &             0.0037           & \textbf{18.7537    }                      & 0.0192 & 0.0025& 0.0217\\ \hline
\multirow{4}{*}{\rotatebox{45}{$n=$\textbf{25}}}   & $0.1n$   & $ 0.9n$     &    63.34     &        0.0222                   & \textbf{63.3622   }                      & 0.1260 &0.0215& 0.1475\\
                         & $0.3n$   & $ 0.7n$    &     47.91       &              0.0165          & 47.9265                       & 0.0695 &0.0145& 0.0840\\
                         & $0.5n$    & $ 0.5n$    &      35.67      &              0.0131       & 35.6831                          & 0.0460 &0.0096& 0.0556\\ 
                         & $0.7n$    & $ 0.3n$   &       23.04     &               0.0083         & \textbf{23.0483}                     & 0.0357 &0.0051& 0.0408\\ \hline
                         
\multirow{4}{*}{\rotatebox{45}{$n=$\textbf{50} }}   & $ 0.1n$  & $ 0.9n$      &  67.24      &             0.0644              & \textbf{67.3044  }                       & 0.2646 &0.0629& 0.3275\\
                         & $0.3n$    & $ 0.7n$   &       52.67     &             0.0499           & 52.7199                          & 0.1957 &0.0457& 0.2414\\
                         & $0.5n$   & $ 0.5n$     &      37.93      &          0.0322           & 37.9622                          & 0.0986 &0.0267& 0.1253\\ 
                         & $0.7n$  & $ 0.3n$     &       21.00     &           0.0227             & \textbf{21.0227 }                        & 0.0798 &0.0172& 0.0970\\ \hline
\multirow{4}{*}{\rotatebox{45}{$n=$\textbf{100}}}  & $0.1n$   & $ 0.9n$     &    96.79      &            0.2002                & \textbf{96.9902   }                       & 0.9700 &0.1968& 1.1668\\
                         & $0.3n$   & $ 0.7n$    &      73.16      &        0.1457                & 73.3057                          & 0.6471 &0.1459& 0.7930\\
                         & $0.5n$  & $ 0.5n$      &       50.02     &        0.0972             & 50.1172                         & 0.4321 &0.0865& 0.5186\\ 
                         & $0.7n$   & $ 0.3n$    &       29.40     &          0.0788              & \textbf{29.4788}                          & 0.2100 &0.0568& 0.2668 \\ \hline
\end{tabular} }
\label{tab:exetime}
\end{table*}
\subsection{Full Numerical Results for \Generate}\label{app:full_num_generate}
Table~\ref{tab:exetime} expands on our experimental results for the \Generate phase and shows the efficiency of each component executed by the clients and the FL server.

\end{document}